\documentclass[11pt]{amsart}
\usepackage{latexsym,amssymb,amsmath,youngtab}
\usepackage{youngtab}
\usepackage{epsfig,url,hyperref}
\usepackage{tikz}%ciken 2013-Sep-10
\usepackage{amsmath,amsthm,amssymb,amscd} %,MnSymbol %ciken 2013-Sep-10   MnSymbol screws up all the symbols!
\usepackage[mathscr]{eucal}
\usepackage{verbatim}
\usepackage{color}
\usepackage{cancel}%ciken 2013-Sep-10 for striking text out
\newtheoremstyle{custom}% name
  {3pt}%      Space above
  {3pt}%      Space below
  {\slshape}%         Body font
  {}%         Indent amount (empty = no indent, \parindent = para indent)
  {\bfseries}% Thm head font
  {.}%        Punctuation after thm head
  { }%     Space after thm head: " " = normal interword space;
   {}%         Thm head spec (can be left empty, meaning `normal')
\theoremstyle{custom}
\numberwithin{equation}{subsection}%ciken 2013-Sep-10 Changed everything to a single counter.
\newtheorem{theorem}[equation]{Theorem}

\newtheorem{proposition}[equation]{Proposition}
\newtheorem{proposition/definition}[equation]{Proposition/Definition}
\newtheorem{lemma}[equation]{Lemma}
\newtheorem{corollary}[equation]{Corollary}
\newtheorem{conjecture}[equation]{Conjecture}

\theoremstyle{definition}
\newtheorem{definition}[equation]{Definition}

\newtheorem{example}[equation]{Example}

\newtheorem{problem}[equation]{Problem}

\theoremstyle{remark}
\newtheorem{remark}[equation]{Remark}

\newtheoremstyle{exercise}% name
  {3pt}%      Space above
  {6pt}%      Space below
  {}%         Body font
  {}%         Indent amount (empty = no indent, \parindent = para indent)
  {\bfseries}% Thm head font
  {:}%        Punctuation after thm head
  { }%     Space after thm head: " " = normal interword space;
   {}%         Thm head spec (can be left empty, meaning `normal')
\theoremstyle{exercise}
\newtheorem{exercise}[equation]{Exercise}
\newtheoremstyle{exercises}% name
  {3pt}%      Space above
  {6pt}%      Space below
  {}%         Body font
  {}%         Indent amount (empty = no indent, \parindent = para indent)
  {\bfseries}% Thm head font
  {:}%        Punctuation after thm head
  {\newline}%     Space after thm head: " " = normal interword space;
   {}%         Thm head spec (can be left empty, meaning `normal')

\theoremstyle{exercise}
\newtheorem{exercises}[equation]{Exercises}

\input epsf
\def\boxit#1{\vbox{\hrule height1pt\hbox{\vrule width1pt\kern3pt
  \vbox{\kern3pt#1\kern3pt}\kern3pt\vrule width1pt}\hrule height1pt}}

\def\trank{\text{rank}}

\def\BC{\mathbb C}
\def\BR{\mathbb R}
\def\BP{\mathbb P}
\def\pp#1{\mathbb P^{#1}}

\def\pp#1{{\mathbb P}^{#1}}
\def\tdim{{\rm dim}}

\def\hd{,...,}

\def\upperp{{}^\perp}

\def\inv{{}^{-1}}

\def\cR{{\mathcal R}}

\def\cL{{\mathcal L}}

\def\cO{{\mathcal O}}

\def\11{\mathbf 1}

\def\fm{{\mathfrak m}}

\def\l{\lambda}
\def\a{\alpha}

\def\o{\omega}

\def\g{\gamma}
\def\s{\sigma}

\def\d{\delta}

\def\ot{{\mathord{ \otimes } }}
\def\op{{\mathord{\,\oplus }\,}}

\def\ra{{\mathord{\;\rightarrow\;}}}

\def\dim{{\rm dim}\;}

\def\frak{\mathfrak}

\def\op{\oplus}
\def\BZ{\Bbb Z}

\def\ep{\epsilon}
\def\op{\oplus}

\def\ul{\underline}

\def\s{\sigma}

\def\a{\alpha}

\def\g{\gamma}
\def\l{\lambda}

\def\FS{\mathfrak  S}

\def\ol{\overline}

\def\BP{\mathbb  P}
\def\BC{\mathbb  C}

\def\pp#1{\mathbb  P^{#1}}
\def\cC{\mathcal  C}

\def\BR{\mathbb  R}

\def\ep{\epsilon}

\def\opc{\op\cdots\op}

\newcommand{\vvirg}{ , \dots , }

\def\hd{, \hdots ,}

\def\inv{{}^{-1}}

\def\pp#1{\mathbb  P^{#1}}

\def\ra{\rightarrow}

\def\tdeg{\operatorname{deg}}

\def\tdim{\operatorname{dim}}

\def\tlim{\lim}

\def\tmin{\operatorname{min}}

\def\trank{\operatorname{rank}}

\def\upperp{{}^{\perp}}

\def\ctimes{\times \cdots\times}

\def\tspan{\operatorname{span}}

\def\be{\begin{equation}}
\def\ene{\end{equation}}

\def\tlog{{\rm{log}}}

\def\tspan{{\rm span}}

\def\G{\Gamma}

\def\tspan{{\rm span}}

\DeclareMathOperator{\Van}{\textit{Van}} %ciken 2013-Jun-17
\DeclareMathOperator{\Vand}{\textit{Vand}} %ciken 2013-Jun-17
\DeclareMathOperator{\Mat}{\textit{Mat}} %ciken 2013-Jun-17
\DeclareMathOperator{\Bfly}{\textit{Bfly}} %ciken 2013-Sep-13

\tikzset{vertex/.style={%ciken 2013-Sep-10
        circle,
        fill=black,
        inner ysep=0pt,
        inner xsep=0pt,
        minimum width = 0.15cm,
        minimum height = 0.15cm}}

\def\trank{{\mathrm {rank}}}

\def\tmult{{\rm mult}}

\def\Om{\Omega}

\def\rig#1{\smash{ \mathop{\longrightarrow}
    \limits^{#1}}}

\def\lc{\cL\cC}

\def\tb{B} % jdh 2013-Sep-22x

\def\rig{\cR ig}

\begin{document}
\sloppy %ciken 2013-Sep-14

\def\BI{{\mathcal I}}\def\BJ{{\bold J}} 
%jml 10-12 these to address CI's concern, font can be changed.

\title{Complexity of linear circuits and geometry}
\author[F. Gesmundo, J.D. Hauenstein, C. Ikenmeyer, J.M. Landsberg]{Fulvio Gesmundo, Jonathan D. Hauenstein, Christian Ikenmeyer and J.M. Landsberg}
%\date{April 2010}
 \begin{abstract}
We use algebraic geometry to study matrix rigidity, and more generally, the complexity of computing
a matrix-vector product, continuing a study initiated in
  \cite{LV,MR2870721}. 
In particular, we (i) exhibit
many  non-obvious equations testing for (border) rigidity, (ii) compute degrees of  varieties
associated to rigidity, (iii) describe
 algebraic varieties associated to families of matrices that are expected to have super-linear rigidity,
 and   (iv) prove   results about the ideals and degrees of cones  that are of interest in their own right.
\end{abstract}
\thanks{Hauenstein supported by NSF grant DMS-1262428 and DARPA Young Faculty Award (YFA)}
\thanks{Landsberg supported by NSF grant DMS-1006353}
\email{fulges@math.tamu.edu, hauenstein@nd.edu, ciken@math.tamu.edu,  jml@math.tamu.edu}
\maketitle

\begin{center} Communicated by Stephen Cook \end{center}

\section{Introduction}
\label{sec:intro}
 Given an $n\times n$ matrix $A$, how many additions  
 are needed to perform the map
\be\label{yesthemap}
x\mapsto Ax,
\ene
where $x$ is a column vector?
L. Valiant  initiated a study of this question in   \cite{MR0660702}.
He used the model of computation of {\it linear circuits} (see \S\ref{lincirsect})
and observed that for a generic linear map one requires a linear circuit of size $n^2$.
He posed the following problem: 
\begin{problem}\label{explicitproblem}  Find  an {\it explicit} sequence of matrices $A_n$
needing   linear circuits of size      super-linear in  $n$ to compute \eqref{yesthemap}.
\end{problem}
\lq\lq Explicit\rq\rq\ has a precise meaning, see \cite{MR2870721}.
Valiant defined a notion of {\it rigidity} that
is a measurement of the size of the best circuit
of a very restricted type~(see \S\ref{lincirsect}) needed to 
compute~\eqref{yesthemap}.  He proved
that strong lower bounds for rigidity implies super-linear lower bounds for any linear circuit
computing \eqref{yesthemap}, see Theorem \ref{valrigthm} below.
This article continues the use of algebraic geometry, initiated in \cite{MR2870721} and the unpublished notes~\cite{LV}, to~study~these~issues.

\subsection{Why algebraic geometry?}
 Given  a polynomial $P$
on the space of $n\times n$ matrices that vanishes on all matrices of low rigidity (complexity), and
a matrix $A$
such that $P(A)\neq 0$, one obtains a lower bound on the rigidity (complexity) of $A$.

For a simple example,
let $\hat \s_{r,n}\subset Mat_{n}$ denote the variety of $n\times n$ matrices of rank at most~$r$.
(If $n$ is understood, we write $\hat\s_r=\hat\s_{r,n}$.)
Then, $\hat \s_{r,n}$ is the zero set of all minors of size $r+1$.  If one minor of size $r+1$ does not vanish on $A$,
we know the rank of $A$ is at least $r+1$.

Define the {\it $r$-rigidity} of an $n\times n$ matrix $M$ to be the smallest $s$ such that $M=A+B$ where $A\in \hat\s_{r,n}$ and $B$ has
 {exactly} %(\emph{otherwise it looks like at least})}
$s$ nonzero entries. Write $\rig_r(M)=s$.

  Define the set of {\it matrices of $r$-rigidity at most $s$}:
\be \label{hcrnrs}
\hat \cR[n,r,s]^0:
 =\{ M\in Mat_{n\times n}\mid \rig_r(M)\leq s\}.
\ene

Thus if we can find a polynomial $P$ vanishing on $\hat \cR[n,r,s]^0$ and a matrix $M$  such that \mbox{$P(M)\neq0$}, we know
$\rig_r(M)>s$. One says $M$ is {\it maximally $r$-rigid} if
$\rig_r(M)=(n-r)^2$, and that $M$ is {\it maximally rigid}
if it is maximally $r$-rigid for all $r$. (See  \S\ref{bbgeomlang} for justification
of this terminology.)

%The main  motivation for the study
%of matrix rigidity is that if a matrix
%$A$  has sufficiently large rigidity, one obtains a super-linear lower bound on the size of any linear circuit computing \eqref{yesthemap}, see Theorem
%\ref{valrigthm} below.
% {[this motivation was already state above. do we want to repeat it? I would remove it from above, and leave it here only.]}

 Our  study has two aspects: finding explicit polynomials vanishing on $\hat \cR[n,r,s]^0$, and proving qualitative information about
such polynomials. The utility of explicit polynomials has already been explained.
For a simple example of a qualitative property,  consider the {\it degree} of a polynomial.
As observed in \cite{MR2870721}, for a given $d$, one can describe matrices that cannot be in the zero set of any polynomial
of degree at most $d$ with integer coefficients. They then give an upper bound $2 n^{2n^2}$ for the degrees of the polynomials
generating the ideal
of the polynomials vanishing on  $\hat \cR[n,r,s]^0$, and describe a  family of matrices that do not satisfy
polynomials of degree $2 n^{2n^2}$ (but this family is not explicit in Valiant's sense).

Following ideas in \cite{LV,MR2870721}, we not only study polynomials related to rigidity, but also to
different classes of matrices of interest, such as Vandermonde matrices. As discussed in
\cite{MR2870721}, one could first try to prove a general Vandermonde matrix is maximally rigid, and then
 afterwards try  to find an explicit sequence of maximally rigid Vandermonde matrices (a problem in $n$ variables instead of $n^2$ variables).

Our results are described in \S\ref{ourresults}. We first recall basic definitions
regarding linear circuits in \S\ref{lincirsect}, give brief descriptions of the relevant
varieties in \S\ref{ourvars},    establish notation in \S\ref{notation},  and describe previous
work in \S\ref{prevwork}.
We have attempted to make this paper readable for both computer scientists and geometers. To this end, we put
off the use of algebraic geometry until~\S\ref{geomsect}, although we use results from it in earlier sections,
and   introduce a minimal amount of geometric language in \S\ref{bbgeomlang}.
We suggest  geometers read \S\ref{geomsect} immediately after  \S\ref{bbgeomlang}.
In \S\ref{idealofjs}, we present our qualitative results about equations.
We give examples of  explicit equations in
\S\ref{exsect}.
We give descriptions of several varieties of matrices  in \S\ref{matvarsect}.
In \S\ref{geomsect}, after reviewing standard facts on joins in \S\ref{gjoinsect},  we present
generalities about the ideals of joins in   \S\ref{idealofjoins},
 discuss   degrees of cones  in \S\ref{degsect} and then apply them to our situation in    \S\ref{degourvars}.

\subsection{Linear circuits}\label{lincirsect}
\begin{definition} A {\it linear circuit}\index{linear circuit}  is a directed acyclic graph $\cL\cC$
in which each directed edge is labeled by a nonzero element of $\BC$. If
$u$ is a  vertex   with incoming edges labeled by $\l_1\hd \l_k$ from
vertices  $u_1\hd u_k$, then $\lc_u$ is the expression
$\l_1 \lc_{u_1}+\cdots +\l_k \lc_{u_k}$.

If $\cL\cC$ has $n$ input vertices  and $m$ output vertices, it determines a matrix
$A_{\lc}\in Mat_{n,m}(\BC)$  by setting
\[
A_i^j := \sum_{{p \text{ path}}\atop{\text{from }i \text{ to } j}} \prod_{{e \text{ edge}}\atop{\text{of } p}} \lambda_e,
\]
and $\cL\cC$ is said to compute $A_{\lc}$.

The {\it size} of   $\lc$ is the number of edges in $\lc$. The {\it depth} of $\lc$ is the length
of a longest path from an input node to an output node.
\end{definition}

Note that size is essentially counting
the number of additions needed to compute $x\mapsto Ax$, so in this model, multiplication by
scalars is \lq\lq free.\rq\rq

For example, the na\"ive algorithm for computing a map $A: \BC^2\ra \BC^3$ gives rise to the complete bipartite graph
as in Figure~\ref{naivefig}.  More generally, the na\"ive algorithm produces a linear circuit of size $O(nm)$.

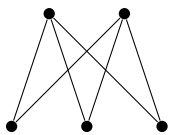
\begin{figure}[!htb]\begin{center}
%ciken 2013-Sep-10
%\includegraphics[scale=.3]{lc1.eps}
\begin{tikzpicture}
\node[vertex] (aI) at (-0.5,1.5) {};
\node[vertex] (bI) at (0.5,1.5) {};
\node[vertex] (aII) at (-1,0) {};
\node[vertex] (bII) at (0,0) {};
\node[vertex] (cII) at (1,0) {};
\draw (aI) -- (aII);
\draw (aI) -- (bII);
\draw (aI) -- (cII);
\draw (bI) -- (aII);
\draw (bI) -- (bII);
\draw (bI) -- (cII);
\end{tikzpicture}
\caption{\small{na\"ive linear circuit for $A\in Mat_{2\times 3}$}}  \label{naivefig}
\end{center}
\end{figure}

If an entry in $A$ is zero, we  may   delete the corresponding edge as in Figure~\ref{edgefig}.

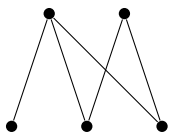
\begin{figure}[!htb]\begin{center}
%ciken 2013-Sep-10
%\includegraphics[scale=.3]{lc3.eps}
\begin{tikzpicture}
\node[vertex] (aI) at (-0.5,1.5) {};
\node[vertex] (bI) at (0.5,1.5) {};
\node[vertex] (aII) at (-1,0) {};
\node[vertex] (bII) at (0,0) {};
\node[vertex] (cII) at (1,0) {};
\draw (aI) -- (aII);
\draw (aI) -- (bII);
\draw (aI) -- (cII);
\draw (bI) -- (bII);
\draw (bI) -- (cII);
\end{tikzpicture}
\caption{\small{linear circuit for   $A\in Mat_{2\times 3}$ with $a^2_1=0$}}  \label{edgefig}
\end{center}
\end{figure}

 Stacking  two graphs $\Gamma_1$ and $\Gamma_2$ on top of each other and
identifying the input vertices of $\Gamma_2$ with the output vertices of $\Gamma_1$,
the matrix of the resulting graph is just the matrix product of the matrices of $\Gamma_1$ and $\Gamma_2$. So, if 
$\trank(A)=1$, we may write $A$ as a product $A=A_1A_2$ where $A_1:\BC^2\ra \BC^1$ and $A_2:\BC^1\ra \BC^3$ and
concatenate the two complete graphs as in Figure 3.

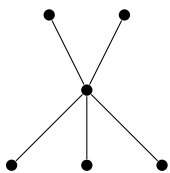
\begin{figure}[!htb]\begin{center}
%ciken 2013-Sep-10
%\includegraphics[scale=.3]{lc2.eps}
\begin{tikzpicture}
\node[vertex] (aI) at (-0.5,2) {};
\node[vertex] (bI) at (0.5,2) {};
\node[vertex] (III) at (0,1) {};
\node[vertex] (aII) at (-1,0) {};
\node[vertex] (bII) at (0,0) {};
\node[vertex] (cII) at (1,0) {};
\draw (aI) -- (III);
\draw (bI) -- (III);
\draw (III) -- (aII);
\draw (III) -- (bII);
\draw (III) -- (cII);
\end{tikzpicture}
\caption{\small{linear circuit for rank one  $A\in Mat_{2\times 3}$}}  \label{twofig}
\end{center}
\end{figure}

 Given two directed acyclic graphs, $\Gamma_1$ and $\Gamma_2$, whose vertex sets are disjoint,
with an ordered list of $n$ input nodes and an ordered list of $m$ output nodes each,
we define the sum $\Gamma_1 + \Gamma_2$ to be the directed graph resulting from
(1) identifying the input nodes of $\Gamma_1$ with the input nodes of $\Gamma_2$,
(2) doing the same for the output nodes,
and (3) summing up their adjacency matrices, see Figure~\ref{fig:sumofgraphs} for an example. 
\begin{figure}[!htb]\begin{center}
%ciken 2013-Sep-10
%\includegraphics[scale=.3]{lc2.eps}
\begin{tikzpicture}
\node[vertex] (aI) at (-0.5,2) {};
\node[vertex] (bI) at (0.5,2) {};
\node[vertex] (III) at (0,1) {};
\node[vertex] (aII) at (-1,0) {};
\node[vertex] (bII) at (0,0) {};
\node[vertex] (cII) at (1,0) {};
\draw (aI) -- (III);
\draw (bI) -- (III);
\draw (III) -- (aII);
\draw (III) -- (bII);
\draw (III) -- (cII);
\end{tikzpicture}
\quad\quad\raisebox{1cm}{+}\quad\quad
\begin{tikzpicture}
\node[vertex] (aI) at (-0.5,2) {};
\node[vertex] (bI) at (0.5,2) {};
\node[vertex] (aII) at (-1,0) {};
\node[vertex] (bII) at (0,0) {};
\node[vertex] (cII) at (1,0) {};
\draw (aI) -- (aII);
\draw (bI) -- (bII);
\end{tikzpicture}
\quad\quad\raisebox{1cm}{=}\quad\quad
\begin{tikzpicture}
\node[vertex] (aI) at (-0.5,2) {};
\node[vertex] (bI) at (0.5,2) {};
\node[vertex] (III) at (0,1) {};
\node[vertex] (aII) at (-1,0) {};
\node[vertex] (bII) at (0,0) {};
\node[vertex] (cII) at (1,0) {};
\draw (aI) -- (III);
\draw (bI) -- (III);
\draw (III) -- (aII);
\draw (III) -- (bII);
\draw (III) -- (cII);
\draw (aI) -- (aII);
\draw (bI) -- (bII);
\end{tikzpicture}
\caption{ {\small{The sum of two graphs.}}}  \label{fig:sumofgraphs}
\end{center}
\end{figure}
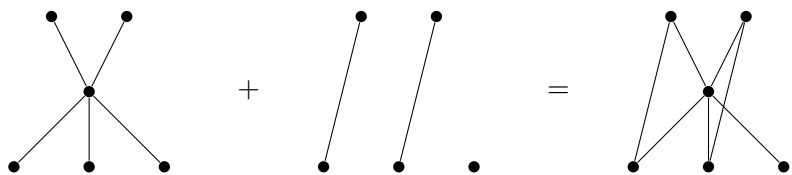

In what follows, for simplicity of discussion, we restrict to the case $n=m$.

 With these descriptions in mind, we see rigidity is a measure of the complexity of a 
 restricted  depth two
circuit   computing  \eqref{yesthemap}. 
Namely the circuit is the sum of two graphs, one of depth one which
has $s$ edges and the other is depth two with $r$ vertices at the middle level.
%It does not appear to be an exact measure,
%because if $S_1,S_2$ are matrices with  $n^{3/2}$ 
%nonzero entries, then $x\mapsto S_1S_2x$ can be computed with a depth two circuit of %size  $2n^{3/2}$ 
%and it is possible that $S_1S_2$ cannot be written as a sum $A+S$
% with \mbox{$2n \trank(A)+s=2 n^{3/2}$}.  
 The motivation for the restriction to such  circuits is Theorem~\ref{valrigthm}.

\subsection{The varieties we study} \label{ourvars}
Define $\hat\cR[n,r,s]:=\ol{\hat\cR[n,r,s]^0}$, the {\it variety of matrices of $r$-border rigidity at most $s$}, where the
overline denotes the common zero set of all polynomials vanishing on
$\hat\cR[n,r,s]^0$, called the {\it Zariski closure}. This equals the   closure of $\hat\cR[n,r,s]^0$ in the classical topology obtained by taking limits,
see  \cite[Thm 2.33]{MR1344216}.
If $M\in \hat \cR[n,r,s]$ we write $\ul{\rig}_r(M)\leq s$, and say $M$ has  {\it $r$-border rigidity at most $s$}.
By definition, $\ul{\rig}_r(M)\leq \rig_r(M)$.
As pointed out in \cite{MR2870721}, strict inequality can occur.
For example, when $s=1$, one obtains points in the tangent cone as in Proposition \ref{joinstdprop}\eqref{as3}.

\smallskip

It is generally expected that there are  super-linear lower bounds
for the size of a linear circuit computing the linear map $x_n\mapsto A_nx_n$ for
 the following sequences of matrices $A_n=(y^i_j)$,  $1 \leq i,j \leq n$, where $y^i_j$ is the entry of $A$ in row $i$ and column $j$:

\smallskip 

{\it Discrete Fourier Transform} (DFT) matrix: let $\o$ be a primitive $n$-th root of unity. Define the size $n$ DFT matrix by  $y^i_j=\o^{(i-1)(j-1)}$.

\smallskip 

{\it Cauchy} matrix: Let $x^i,z_j$ be variables $1\leq i,j\leq n$, and define   $y^i_j=\frac 1{x^i+z_j}$.
(Here and in the next example, one means super linear lower bounds for a sufficiently  general assignment of the variables.)

\smallskip

{\it Vandermonde} matrix: Let $x_i$, $1\leq i\leq n$,  be  variables, define   $y^i_j:=(x_j)^{i-1}$.

\smallskip

{\it Sylvester} matrix: $Syl_1=\begin{pmatrix} 1&1\\ 1&-1\end{pmatrix}$, $Syl_k=\begin{pmatrix} S_{k-1}&S_{k-1}\\ S_{k-1}&-S_{k-1}\end{pmatrix}$.

\smallskip

We describe algebraic varieties associated to classes of matrices generalizing these examples, describe their ideals and make
basic observations about their rigidity.

To each directed acyclic graph $\G$ with $n$ inputs and outputs, or sums of such, we may associate a variety
$\Sigma_{\G}\subset Mat_n$ consisting of the closure of all matrices $A$ such that \eqref{yesthemap} is computable by $\G$.
For example, to the graph in Figure \ref{figfour}
we associate the variety $\Sigma_{\G} := \hat\s_{2,4}$
 since any $4\times4$ matrix of rank at most $2$ can be written a product of a $4\times2$ matrix and a $2\times4$ matrix.

\begin{figure}[!htb]\begin{center}
%ciken 2013-Sept-10
%\includegraphics[scale=.3]{rk24x4.eps}
\begin{tikzpicture}
\node[vertex] (aI) at (-1.5,2) {};
\node[vertex] (bI) at (-0.5,2) {};
\node[vertex] (cI) at (0.5,2) {};
\node[vertex] (dI) at (1.5,2) {};
\node[vertex] (aIII) at (-0.5,1) {};
\node[vertex] (bIII) at (0.5,1) {};
\node[vertex] (aII) at (-1.5,0) {};
\node[vertex] (bII) at (-0.5,0) {};
\node[vertex] (cII) at (0.5,0) {};
\node[vertex] (dII) at (1.5,0) {};
\draw (aI) -- (aIII);
\draw (aI) -- (bIII);
\draw (bI) -- (aIII);
\draw (bI) -- (bIII);
\draw (cI) -- (aIII);
\draw (cI) -- (bIII);
\draw (dI) -- (aIII);
\draw (dI) -- (bIII);
\draw (aIII) -- (aII);
\draw (bIII) -- (aII);
\draw (aIII) -- (bII);
\draw (bIII) -- (bII);
\draw (aIII) -- (cII);
\draw (bIII) -- (cII);
\draw (aIII) -- (dII);
\draw (bIII) -- (dII);
\end{tikzpicture}
\caption{\small{linear circuit for rank two  $A\in Mat_{4}$}}  \label{figfour}
\end{center}
\end{figure}
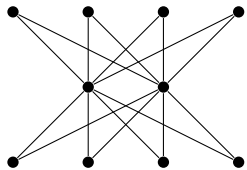
 Note that the number of edges of $\G$ gives an upper bound of
the dimension of $\Sigma_{\G}$, but the actual dimension is often less, for example
$\tdim\hat\s_{2,4}=12$
but $\G$ has $16$ edges. This is because there are four parameters of
choices for expressing a rank two matrix
as a sum of two rank~one~matrices.

\subsection{Notation and conventions}\label{notation} Since this article is for geometers and computer scientists,
here and throughout, we include  a substantial amount of material  that is not usually~mentioned.

We work exclusively over the complex numbers $\BC$.

For simplicity of exposition, we generally restrict to square matrices, although most results carry over
to rectangular matrices as well.

Throughout, $V$ denotes a complex vector space, $\BP V$ is  the associated projective
space of lines through the origin in $V$, $S^dV^*$ denotes the space of homogenous polynomials of degree $d$ on~$V$,
and $Sym(V^*)=\oplus_d S^dV^*$ denotes the symmetric algebra, i.e.,
the ring of polynomials on $V$, i.e, after a choice of basis,
the ring of polynomials in $\tdim V$ variables. We work with projective space because
the objects of interest are invariant under rescaling and to take advantage of results in projective
algebraic  geometry, e.g., Proposition \ref{conedegprop}.   For a subset  $Z\subset \BP V$,   $\hat Z:=\pi\inv(Z)\cup \{0\}\subset V$  is called  the  {\it affine  cone~over~$Z$}.

Let  $Z\subset \BP V$ be a {\it projective variety}, the zero set of a collection of homogeneous
polynomials on $V$ projected to $\BP V$. The ideal of $Z$, denoted $\BI(Z)$,  is the ideal in $Sym(V^*)$
of all polynomials vanishing on $\hat Z$.   Let $\BI_d(Z)\subset S^dV^*$ denote
the degree $d$ component of the ideal of $Z$.
 The {\it codimension} of $Z$ is the smallest non-negative integer $c$ such that every linear $\BP^c\subset \BP V$
intersects~$Z$ and its {\it dimension} is $\tdim \BP V-c$. 
The {\it degree} of  $Z$  is the number of points of intersection
with a general linear space of dimension $c$. 
Here and throughout, a {\it general point} or {\it general linear space} is a point
(or linear space) that does not satisfy certain (specific to the problem) polynomials,
so the set of general points is of full measure, and one may  simply view 
  a general point or linear space as one that has been randomly chosen. 
 A codimension $1$ variety is called a {\it hypersurface}
and is defined by a single equation.  The degree of a hypersurface is the degree
of its defining equation. 

For a linear subspace $U\subset V$, its annihilator in the dual space is denoted $U\upperp \subset V^*$, and
we abuse notation and write $(\BP U)\upperp \subset V^*$ for the annihilator of $U$ as well.
The group of invertible endomorphisms of $V$ is denoted $GL(V)$. If   $G\subset GL(V)$ is a subgroup and
$Z\subset \BP V$ is a subvariety such that $g\cdot z\in Z$ for all $z\in Z$ and all $g\in G$, we say
$Z$ is a {\it  $G$-variety}. The group of permutations on $d$ elements is denoted $\FS_d$.

%Repeated indices appearing up and down are to be summed over
%and 
We   write $\tlog$ to mean $\tlog_2$. 

Let $f,g: \BR \ra \BR$ be functions. Write $f= \Om(g)$ (resp. $f= O(g)$) if and only if there exists $C>0$ and $x_0$ such that
$|f(x)|\geq C|g(x)|$ (resp. $|f(x)|\leq C|g(x)|$)  for all $x\geq x_0$.   Write $f=\o(g)$  (resp. $f=o(g)$)
if and only if for all  $C>0$ there exists   $x_0$ such  that
$|f(x)|\geq C|g(x)|$ (resp. $|f(x)|\leq C|g(x)|$)  for all $x\geq x_0$.
These definitions are used for any ordered range and domain, in particular $\BZ$.
In particular, for a function $f(n)$,   $f=\o(1)$ means $f$  goes to infinity as $n\ra \infty$.

For $I,J\subset[n] :=  \{1,2,\ldots,n\}$ of size $r+1$,
 {let} $M^I_J$  {be} the determinant of the size $r+1$ submatrix defined by $(I,J)$.
Set $\Delta^I_J:=M^{I^c}_{J^c}$, where   $I^c$ and $J^c$ denote the complementary index set to $I$ and $J$, respectively.  We often use $x^i_j$ to denote coordinates
on the space of $n\times n$ matrices. Write $\{ x^i_j\}:= \{ x^i_j :  {i,j \in [n]}\}$.  

\subsection{Previous Work}\label{prevwork}

The starting point is the following theorem of L. Valiant:
\begin{theorem} \cite[ Thm.~6.1, Prop.~6.2]{MR0660702}\label{valrigthm}
 Suppose that a sequence   $A_n\in Mat_{n}$ admits  a sequence of linear circuits of size $\Sigma=\Sigma(n)$ and depth $d=d(n)$ where each gate
has fan-in two. Then for any $t>1$,
$$
\rig_{\frac{\Sigma\tlog(t)}{\tlog(d)}}(A_n)\leq 2^{O(d/t)}n.
$$
 In particular, if there exist $\ep,\d>0$ such that $\rig_{\ep n}(A_n)=\Omega( n^{1+\d})$, then any sequence of linear circuits
of logarithmic (in $n$) depth computing $\{A_n\}$ must have size $\Om(n\tlog(\tlog n))$. 
\end{theorem}
\medskip

\begin{proposition} (\cite{MR1237045} for finite fields and
\cite{MR1608240} for the general case)  Let $r\geq (\tlog n)^2$, and let $A\in Mat_{n\times n}$ be such that all minors
of size $r$ of $A$ are nonzero. Then, for all $s<\frac{n^2}{4(r+1)}\tlog(\frac nr)$,  $A\not\in \hat\cR[n,r,s]^0$. \label{pro:minorsnonzero}
\end{proposition}

Note that if one sets $r=\ep n$, for $n$ sufficiently large, the above result says
$A\not\in \hat\cR[n,\ep n,\frac 1{\ep^2}\tlog(\frac 1\ep)]^0$
which is far from what would be needed to apply Theorem \ref{valrigthm},  as $s$
does not grow with $n$.
 
The matrices $DFT_p$ with $p$ prime, general Cauchy matrix, general Vandermonde matrix, general Sylvester matrix are
such that all minors of all sizes are nonzero (see  \cite[\S 2.2]{MR2539154} and the references therein).
 Thus Proposition~\ref{pro:minorsnonzero} implies: 

\begin{corollary}  \cite{MR1237045, MR1608240,MR1680943} The   matrices of the following types:
 $DFT_p$ with $p$ prime,   Cauchy,  Vandermonde,   and Sylvester,
are such that  for all $s<\frac{n^2}{4(r+1)}\tlog(\frac nr)$,  $A\not\in \hat\cR[n,r,s]^0$.
\end{corollary}

The following theorem is   proved via a theorem in graph theory from \cite{MR0446750}:

\begin{theorem} (attributed to Strassen in \cite{MR0660702}, also see \cite[\S 2.2]{MR2539154})
 For all $\ep > 0$, 
there exist $n\times n$ matrices $A$
with integer entries, all of whose minors of all  sizes   are nonzero such that
$A\in \hat \cR[n,\ep n,n^{1+o(1)}]^0$. %for any $\ep>0$.
\end{theorem}

In \cite{MR2870721}, they approach the rigidity problem from the perspective of algebraic geometry. In particular, they
use the effective Nullstellensatz to obtain bounds on the degrees of the hypersurfaces of maximally border rigid matrices.
They show the following.

\begin{theorem} \cite[Thm.~4.4]{MR2870721} \label{thm7}
Let $p_{k,j}>2n^{2n^2}$ be distinct primes for $1\leq k,j\leq n$. Let $A_n$ have entries $a^k_j=e^{2\pi i/p_{k,j}}$.
Then $A_n$ is maximally $r$-border  rigid for all $1\leq r\leq n-2$.
\end{theorem}

See Remark \ref{improverem}  for a  small  improvement of this result.

In \cite{MR2870721}, they do not restrict their field to be $\BC$.

Additional references for matrix rigidity are
\cite{MR1734200,MR1892859,MR2277251,MR2305515,MR2822875,MR1137781,MR1756225,MR2436904}.

\subsection{Our results}\label{ourresults}

Previous to our work, to our knowledge, there were no explicit equations for irreducible components of  $\hat\cR[n,r,s]$ known
other than the minors of size $r+1$. The irreducible components  of  $\hat\cR[n,r,s]$ are determined (non-uniquely) by cardinality $s$ subsets $S\subset
\{ x^i_j\}$      corresponding to the entries one is allowed to change.
Recall that $x^i_j$ are coordinates on the space of $n\times n$ matrices. 
 We  find   equations for  the     cases $r=1$, (Proposition \ref{thm: hypersurface for r=1}),  $r=n-2$
(Theorem \ref{rnmtwothm}), and the cases $s=1,2,3$ (see \S\ref{firstexamples}). 
We also obtain qualitative information about the equations. Here are some sample results:

\begin{proposition}\label{qualprop} Each irreducible component of $\hat \cR[n,r,s]$,
described by some set \mbox{$S\subset\{x^i_j, 1 \leq 1,j \leq n\}$},  has ideal generated by polynomials with the following property;
% \begin{enumerate}
% \item For each generator $P$ of degree $d$, there exist 
% multi-indices  $I,J: [d]\ra  [n]$  of   cardinality  $d$,  and a subset $\Sigma\subset \FS_d$
% such that all monomials appearing in the expression of $P$ are precisely of the form
% $$
% x^{i_1}_{j_{\s(1)}}\cdots x^{i_d}_{j_{\s(d)}}
% $$
% for $\s\in \Sigma$.
if $P$ is a generator of degree $d$, then no entries of $S$ appear in $P$ and $P$ is a sum of
terms of the form $\Delta Q$ where $\Delta$ is a minor of size $r+1$ and $\tdeg(Q)=d-r-1$. In particular, there are
no equations of degree less than $r+1$ in the ideal.

Conversely any polynomial $P$ of degree $d$ such that no entries of $S$ appear in $P$ and $P$ is a sum of
terms $\Delta Q$ where $\Delta$ is a minor of size $r+1$ is in the ideal of the component of $\hat\cR[n,r,s]$
determined by $S$.
%\end{enumerate}
\end{proposition}
See \S\ref{idealofjs} for more precise statements. These results
are consequences of more general results   about cones in  \S\ref{gjoinsect}.

We remind the reader that $\Delta^I_J$ is the determinant of the submatrix obtained by deleting the rows of $I$ and the columns of $J$.

\begin{theorem}  There are two types of components
of the hypersurface $\hat\cR[n,n-2,3]$:
\begin{enumerate}
\item Those corresponding to a configuration $S$ where the three entries  are all
in distinct rows and columns, where if   $S=\{ x^{i_1}_{j_1},x^{i_2}_{j_2},x^{i_3}_{j_3}\}$
the
  hypersurface  is  of degree $2n-3$ with equation
$$
\Delta^{i_3}_{j_2}\Delta^{i_1,i_2}_{j_1,j_3}-
\Delta^{i_2}_{j_3}\Delta^{i_1,i_3}_{j_1,j_2}=0.
$$
\item Those corresponding to a configuration where there are
 two elements of $S$ in the same row and one in a different column  from those two,
or such that
one element shares a row with one and a column with the other. In these cases,  the equation
is the unique size $(n-1)$ minor that has no elements of $S$.
\end{enumerate}

If all three elements of $S$  lie on a row or column, then one does not obtain a hypersurface.
\end{theorem}

 We give numerous examples of equations in other special cases in \S\ref{exsect}.
Our main tool for finding these equations are the results  {presented in} \S\ref{idealofjs}, which follow from more
general results regarding joins of projective varieties that we prove in \S\ref{idealofjoins}.

If one holds not just $s$ fixed, but moreover fixes the specific entries of the matrix that one is allowed
to change, and allows the matrix to grow (i.e., the subset $S$ is required to be contained in  some $n_0\times n_0$ submatrix of
$A\in Mat_n$), there is  a    {\it propagation} result (Proposition \ref{degboost}),
that enables one to deduce the equations in the $n\times n$ case from the $n_0\times n_0$ case.

 When one   takes a cone
 (in the sense of Definition \ref{joindef} below, not to be confused with the
 affine cone over a variety)  over a variety 
 with respect to  a general linear space, there is a dramatic increase in the degree
 of the generators of the ideal 
because the equations of the cone are obtained using  
elimination theory. For example,  a general cone over a codimension two complete intersection, whose ideal is generated in
degrees $d_1,d_2$ will have defining equation in  degree $d_1d_2$. However, we are taking cones over very
singular points of varieties that initially are not complete intersections, so the increase in degree is significantly less. We conjecture: 

\begin{conjecture}\label{slowgrowconj} Fix $0<\ep<1$ and $0<\d<1$. Set $r=\ep n$ and $s=n^{1+\d}$. Then the minimal degree of
a polynomial in the ideal of each irreducible component of  $\hat\cR[n,r,s]$ grows like a polynomial in $n$.
\end{conjecture}

Although it would not immediately solve Valiant's problem, an affirmative answer to Conjecture \ref{slowgrowconj} would
drastically simplify the study.

 While  it is difficult to get direct information about the degrees of defining
equations of the irreducible components of  $\hat\cR[n,r,s]$, as na\"\i vely one needs to use elimination theory,
one can use general results from algebraic geometry to get information about the degrees of the varieties.

Let $d_{ {n,r},s}$ denote the 
maximum degree of an irreducible  component of $\hat\cR[n,r,s]$.
 It will be useful to set $k=n-r$. 
Then   (see e.g.,  \cite[p95]{MR770932} for the first equality and e.g. \cite[p. 50,78]{FH} for the fourth and fifth)
\begin{align}\label{drno}
d_{ {n,r},0}&= \prod_{i=0}^{n-r-1}\frac{(n+i)!i!}{(r+i)!(n-r+i)!}\\
\nonumber &=
\frac{\tb(r) \tb(2n-r)\tb(n-r)^2}{\tb(n)^2\tb(2n-2r)}\\
\nonumber &=
\frac{\tb(n-k) \tb(n+k)\tb(k)^2}{\tb(n)^2\tb(2k)}\\
\nonumber &
=\tdim S_{k^{k}}\BC^n\\
&=\frac{\tdim [k^k]}{k^2!}\frac{\tb(n-k) \tb(n+k) }{\tb(n)^2}
\end{align}
Here $\tb(k):=G(k+1)$, where $G(m)=\prod_{i=1}^{m-2}i!$ is the Barnes $G$-function, 
  $S_{k^{k}}\BC^n$ denotes  the irreducible $GL_n$-representation of type $(k,k,\ldots,k)$, and  
  $[k^k]$ denotes the irreducible $\FS_{k^2}$-module corresponding
to the partition $(k\hd k)$.

\begin{remark}\label{improverem}
The shifted Barnes $G$-function $\tb$  
has the following asymptotic expansion
%$$
%\tln (\tb(z))=z^2(\frac 1{2}\tln(z)-3/4)+\frac 12 \tln (2\pi)z-\frac 1{12}\tln(z)+{\zeta'}(-1)+O(\frac 1z)
%$$
%thus
$$
\tb(z) = \left(\frac z{e^{\frac 32}}\right)^{\frac{z^2}2}O(2.51^z)
$$
(see e.g. \url{en.wikipedia.org/wiki/Barnes_G-function}).
 Since the degree of a variety cannot increase  when taking a cone over it, 
  one can replace the $2n^{2n^2}$ upper bound  in Theorem \ref{thm7} with  roughly $n^{ \ep n^2}$
  because, setting
$r=\ep n$,  for some constant $C$, 
$$
 \frac{\tb(\ep n) \tb((2-\ep)n )\tb((1-\ep)n)^2}{\tb(n)^2\tb(2(1-\ep)n)}
\leq  n^{n^2[\frac {\ep^2}2+\frac{(2-\ep)^2}2 +(1-\ep)^2 -1-2(1-\ep)^2]}C^{n^2+n}
=n^{\ep n^2}C^{n^2+n}
$$
\end{remark}

\begin{remark}
A geometric interpretation
of the equality between  $\tdeg d_{n,r,0}$ and  the dimension of an irreducible $GL_n$-module is discussed in \cite{MR2116717}.
\end{remark}

We prove
several   results  about the degrees $d_{n,r,s}$. For example:

\begin{theorem}\label{bbdegthrm} Let $s\leq n$, 
Then,
\be\label{degsum}
d_{ n,r ,s}\leq d_{ {n,r},0}-\sum_{j=1}^s d_{ {n-1,r-1},s-j}
\ene
\end{theorem}
 
In an earlier version of this paper, we conjectured that equality held in
\eqref{degsum}.
After we submitted the paper for publication, our conjecture
was answered affirmatively in \cite{aluffideg}:

\begin{theorem}\label{bbdegthrmaluffi}\cite{aluffideg} Let $s\leq n$, 
Then,
\be\label{degsuma}
d_{ n,r ,s}= d_{ {n,r},0}-\sum_{j=1}^s d_{ {n-1,r-1},s-j}
\ene
\end{theorem}

In the previous version of this paper, the following theorem
was stated 
with the hypothesis that   equality holds in \eqref{degsum} for all $(r',n',s')\leq (r,n,s)$  and $s\leq n$.
Theorem \ref{bbdegthrmaluffi} renders it 
to the present unconditional form:

\begin{theorem} Each irreducible component of $\hat \cR[n,n-k,s]$ has degree at most
\be \label{amazingsum}
\sum_{m=0}^s\binom sm (-1)^m d_{r-m,n-m,0}
\ene
with equality holding if  {no two elements of} $S$  {lie in the same row or column,} 
e.g., if the elements of $S$ appear on the diagonal.

Moreover, if we set $r=n-k$ and $s=k^2-u$ and consider the degree  {$D(n,k,u)$}
as a function of $n,k,u$, then,   fixing $k,u$ and considering $D_{k,u}(n)=D(n,k,u)$ as a function of $n$, 
it is of the form
$$
D_{k,u}(n)= (k^2)!\frac{\tb(k)^2}{\tb(2k)}p(n)
$$
where $p(n)=\frac{n^u}{u!}- \frac{k^2-u}{2(u-1)!}n^{u-1} +O(n^{u-2})$ is a polynomial of degree $u$. 
\end{theorem}

For example, when $u=1$,  
$ 
D(n,k,1)=  
 (k^2)!\frac{\tb(k)^2}{\tb(2k)} (n- \frac 12 (k^2-1)). 
$

\begin{remark}Note that  $D_{k,u}(n)= \tdim [k^k]p(n)$. It would be nice to have a geometric or representation-theoretic explanation
of this equality.
\end{remark}

 \smallskip
 
 \begin{remark} In our earlier version of this paper, we realized that the use
 of 
  {\it intersection
theory} (see, e.g. \cite{FultonIT}), could render Theorem \ref{bbdegthrm} unconditional,
so we contacted P. Aluffi, an expert in the subject. 
Not only was he able to render the theorem unconditional, but he determined the
degrees in additional cases. We are delighted that such beautiful geometry can
be  of use to computer science, and look forward to further progress on these
questions.
We expect a substantial reduction in degree when $r=\ep n$ and
$s=(n-r)^2-1$. 
\end{remark}

\medskip

We  define varieties modeled on different classes of families
of matrices as mentioned above. We show that a general Cauchy matrix, or
a general Vandermonde matrix is maximally $1$-rigid and maximally $(n-2)$-rigid (Propositions \ref{cauchyrig}
and \ref{vandrigid}).
One way to understand the
DFT algorithm is to factor the  discrete Fourier transform matrix as a product (set $n=2^k$)
$DFT_{2^k}=S_1\cdots S_k$ where each $S_k$ has only $2n$ nonzero entries. Then, these sparse
matrices can all be multiplied via  a linear circuit of size $2n\tlog n$ (and depth $\tlog n$).
We define the variety of {\it factorizable} or {\it butterfly}  matrices $FM_n$ to be the closure of the set of   matrices admitting such a
description as a product of sparse matrices, all of which admit
a linear circuit of size $2n\tlog n$, and show (Proposition \ref{factorprop}):
\begin{proposition}
A general butterfly matrix admits a linear circuit of size $2n\tlog n$, but does not admit a linear
circuit of size $n(\tlog n+1)-1$.
\end{proposition}

\subsection{Future work}   Proposition  \ref{qualprop} gives    qualitative information about the ideals and
we give numerous examples of equations for the relevant varieties.
It would be useful to continue the study of the equations both qualitatively and by computing further
explicit examples,  with the hope of eventually getting  
equations in the Valiant range. In a different direction, an analysis of the degrees of the hypersurface cases
in the range $r=\ep n$ could lead to a substantial reduction of the known degree bounds.

Independent of  complexity theory, several interesting questions relating the differential geometry
and scheme structure of tangent cones are posed in \S\ref{geomsect}.

\subsection*{Acknowledgements} We thank the anonymous referees for   very
careful reading and numerous useful suggestions. 

\section{Geometric formulation}

\subsection{Border rigidity}\label{bbgeomlang}

\begin{definition}\label{joindef}  For varieties $X,Y\subset \BP V$, let
$$\BJ^0(X,Y):= \bigcup_{x\in X, y\in Y, x\neq y} \langle x,y\rangle
$$
and  {define the {\it join of $X$ and $Y$} as} $\BJ(X,Y):=\ol{\BJ^0(X,Y)}$
 {with closure using either the Zariski or the classical topology}.
Here, $\langle x,y\rangle$ is the (projective) linear span of the points $x,y$.
If $Y=L$ is a linear space $\BJ(X,L)$ is called the {\it cone} over $X$ with {\it vertex $L$}.
(Algebraic geometers refer to $L$ as the vertex even when it is not just
a point $\pp 0$.)
\end{definition}

Let $\s_r=\s_{r,n}=\s_r(Seg(\pp{n-1}\times \pp{n-1}))  \subset \BP (\BC^n\ot \BC^n)$ denote the variety of  up to scale  {$n\times n$ matrices of}
rank at most $r$, called the {\it $r$-th secant variety of the Segre variety}.
For those not familiar with this variety, the Segre variety itself, $Seg(\pp{n-1}\times \pp{n-1})$, is the projectivization
of the rank one matrices, one may think of the first $\pp{n-1}$ as parametrizing
column vectors and the second as parametrizing row vectors and
the corresponding (up to scale) rank one matrix as their product. The rank at most
$r$-matrices are those appearing in some secant $\pp{r-1}$ to the 
Segre variety. 

Let $\BC^n\ot \BC^n$ be furnished with linear coordinates $x^i_j$, $1\leq i,j\leq n$.
Let $S\subset \{ x^i_j\}$ be a subset
of cardinality~$s$ and let $L^S:= \tspan~S$.
We may rephrase \eqref{hcrnrs} as
$$
\hat \cR[n,r,s]^0 =\bigcup_{S\subset \{ x^i_j\} , |S|=s } \hat \BJ^0(\s_r,L^S)
$$
 The  dimension of $\s_r$ is $r(2n-r)-1$ and   $\tdim \hat \BJ(\s_r,L^S)\leq \tmin\{ r(2n-r)+s, n^2\}$ (see Proposition \ref{joinstdprop}(3)). 
We say the dimension is the {\it expected dimension}  if equality holds.

The variety $\hat \cR[n,r,s]$, as long as $s>0$ and it is not the ambient space,  is reducible, with at most ${\binom {n^2} s}$ components,
all of the same dimension $r(2n-r)+s$.
(This had been observed in \cite[Thm. 3.8]{MR2870721} and \cite{LV}.)
To see the equidimensionality,
notice that if $\vert S_j \vert = j$ and $S_j \subseteq S_{j+1}$, then the sequence of 
joins $\BJ_j = \BJ(\sigma_r,L^{S_j})$ eventually fills the ambient space. Moreover, $\dim \BJ_{j+1} \leq \dim \BJ_j +1$, 
so the only possibilities are $\BJ_{j+1} = \BJ_j$ or $\dim \BJ_{j+1} = \dim \BJ_j +1$. In particular, this shows that for 
any $j$ there exists a suitable choice of $S_j$ such that $\BJ_j$ has the expected dimension.  
  Now, suppose that $\BJ(\sigma_r,L^S)$ 
does not have the expected dimension, so its dimension is $r(2n-r)+s'$ for some $s' < s$. Let $S' \subseteq S$ be such that
$\vert S'\vert =s'$ and $\BJ(\sigma_r,L^{S'})$ has the expected dimension. Then $\BJ(\sigma_r,L^{S'}) = \BJ(\sigma_r,L^{S})$. 
Now let $R$ be such that $\vert R \vert = s$, $S' \subseteq R$ and $\BJ(\sigma_r, L^R)$ has the expected dimension. Then $\BJ(\sigma_
r, L^R)$ is an irreducible component of $\hat \cR[n,r,s]$ that  contains $\BJ(\sigma_r,L^{S})$, showing  that the 
irreducible components of $\hat \cR[n,r,s]$ have dimension equal to the expected dimension of $\BJ(\sigma_r,L^S)$.
Therefore (again by \cite[Thm. 3.8]{MR2870721} and \cite{LV}), 
\be\label{expdimR}\tdim \hat \cR[n,r,s]=\tmin\{ r(2n-r)+s, n^2\}, 
\ene
and  $\hat \cR[n,r,s]$ is a hypersurface if and only if
\be\label{hypers}
s=(n-r)^2-1.
\ene
We say  a matrix $M$ is {\it maximally $r$-border rigid} if $M\not\in \hat\cR[n,r,(n-r)^2-1]$, and that
$M$ is {\it maximally border rigid} if $M\not\in \hat \cR[n,r,(n-r)^2-1]$ for all $r=1\hd n-2$.
Throughout we assume $r\leq n-2$ to avoid trivialities.

The set of maximally  
rigid matrices is of full measure (in any reasonable measure e.g. absolutely continuous with respect to Lebesgue measure)
on the space of $n\times n$ matrices. In particular, a \lq\lq random\rq\rq\ matrix   will be maximally
rigid.  

\subsection{On the ideal of $\BJ(\s_r,L^S)$}\label{idealofjs}
Write $S^c=\{ x^i_j \}\backslash S$ for the complement of $S$.
The following is a consequence of Proposition \ref{jjens}:

\begin{proposition}\label{usefulprop} Fix $L=L^S$. Generators for the ideal of $\BJ(\s_r,L)$ may be obtained from
polynomials of the form
$P =\sum_{I,J}q_I^{J} M^I_{J}$, where  the  $q_I^J$ are arbitrary homogeneous
polynomials all of the same degree and:
\begin{enumerate}
\item $M^I_{J}$ is the (determinant of the) size $r+1$ minor defined by the index sets $I,J$
(i.e., $I,J\subset [n]$, $|I|=|J|=r+1$), and
\item  only the variables of $S^c$ appear in $P$.
\end{enumerate}

Conversely, any polynomial of the form $P =\sum_{I,J} q_I^{J} M^I_{J}$, 
where the $M^I_{J}$ are
  minors of size $r+1$ and    only the variables of $S^c$ appear in $P$,  is in $\BI(\BJ(\s_r,L))$.
\end{proposition}

Let $E,F=\BC^n$.  The irreducible polynomial representations
of $GL(E)$ are indexed by partitions $\pi$ with at most $\tdim E$ parts. 
Let $\ell(\pi)$ denote the number of parts of $\pi$, and let $S_{\pi}E$ denote the irreducible $GL(E)$-module
corresponding to $\pi$. We have the $GL(E)\times GL(F)$-decomposition 
$$S^d(E\ot F)= \bigoplus_{|\pi|=d,\, \ell(\pi)\leq n} S_{\pi}E\ot S_{\pi}F.
$$
Let $T_E\subset GL(E)$ denote  the torus (the   invertible diagonal matrices).
A vector $e\in E$ is said to be a {\it weight vector} if $[t\cdot e]=[ e]$ for all $t\in T_E$.

\begin{proposition}\label{weightprop}  Write $Mat_n=E\ot F$. For all $S\subset \{ x^i_j\}$, $\BJ(\s_r,L^S)$ is a $T_E\times T_F$-variety.

Thus a set of generators of $\BI(\BJ(\s_r,L))$ may be taken from $GL(E)\times GL(F)$-weight vectors and these weight
vectors must be sums of vectors in modules $S_{\pi}E\ot S_{\pi}F$ where $\ell(\pi)\geq r+1$.
\end{proposition}

The length requirement follows from Proposition \ref{usefulprop}(1).

Proposition \ref{qualprop}(1)  is Proposition \ref{weightprop} expressed in coordinates.
For many examples, the generators   have nonzero projections onto all the modules $S_{\pi}E\ot S_{\pi}F$ with $\ell(\pi)\geq r+1$.

Recall  the notation $\Delta^{I}_{J}=M^{I^c}_{J^c}$, where $I^c$ denotes the complementary index set to $I$. This
will allow us to work independently  of the size of our matrices.

   Let $S\subset \{x^i_j\}_{1\leq i,j\leq n}$ and let $P\in \BI_d(\s_r)\cap S^d(L^S)\upperp$, so $P\in \BI_d(\BJ(\s_r,L^S))$,
and require further that $P$ be a $T_E\times T_F$ weight vector. Write 
\be\label{pexpr}
P=\sum_v \Delta^{I_1^v}_{J_1^v}  \cdots \Delta^{I_f^v}_{J_f^v}
\ene
where $|I_{\a}^1|=|J_{\a}^1| =: \d_{\a}$,  so
$d=fn-\sum_{\a}\d_{\a}$. Since we allow $\d_{\a}=1$,
any polynomial that is a weight vector may be written in this way.
Write $x^i_j=\ol{e}_i\ot \ol{f}_j$.  
 The $T_E$-weight of  $P$ is $ (1^{\l_1}\hd n^{\l_n})$ where $\ol{e}_j$ appears $\l_j$ times in the union of the   $(I^v_{\a})^c$'s, and the $T_F$-weight of $P$
 is  $ (1^{\mu_1}\hd n^{\mu_n})$ where $\ol{f}_j$ appears $\mu_j$ times in  the union of the   $(J^v_{\a})^c$'s. 

Define
\be\label{pqexpr}
P_q\in S^{d+qf}Mat^*_{n+q }
\ene
by \eqref{pexpr} only considered as a polynomial on  $Mat_{n+q}$, and note that   $P=P_0$.

\medskip

\begin{proposition} \label{degboost} For $P\in \BI_d(\s_{r,n})\cap S^d(L^S)\upperp$ as in \eqref{pexpr},
$$P_q\in \BI_{d+fq} (\s_{r,n+q})\cap S^{d+fq}(L^S)\upperp
$$ where $S$ is the same for $Mat_n$ and $Mat_{n+q}$. In particular, $P_q\in \BI_{d+fq} (\BJ(\s_{r+q,n+q},L^S))$.
\end{proposition}

\begin{proof}
It is clear $P_q\in \BI_{d+fq} (\s_{r,n+q})$, so it remains to show it is in $S^{d+fq}(L^S)\upperp$.
By induction, it will be sufficient to prove the case $q=1$.
Say in some term, say $v=1$,  in the summation of $P$ in \eqref{pexpr} a monomial in $S$ appears as a factor,
some
$x^{s_1}_{t_1}\cdots x^{s_g}_{t_g}Q $.
Then, by Laplace expansions, we may write $Q =\tilde \Delta^{I_1^1}_{J_1^1}  \cdots \tilde \Delta^{I_f^1}_{J_f^1}$,
for some minors (smaller than or equal to the originals).
Since this term is erased we must have, after re-ordering terms, for $v=2\hd h$ (for some $h$),
$$
x^{s_1}_{t_1}\cdots x^{s_g}_{t_g}(\tilde \Delta^{I_1^1}_{J_1^1}  \cdots \tilde \Delta^{I_f^1}_{J_f^1}
+ \cdots +
\tilde \Delta^{I^h_1}_{J^h_1}  \cdots \tilde \Delta^{I_f^h}_{J_f^h})=0
$$
that is,
\be\label{deleqn}
\tilde \Delta^{I_1^1}_{J_1^1}  \cdots \tilde \Delta^{I_f^1}_{J_f^1}
+ \cdots +
\tilde \Delta^{I_h^1}_{J_h^1}  \cdots  \tilde \Delta^{I_f^h}_{J_f^h}=0.
\ene
Now consider the same monomial's appearance in $P_1$ (only the monomials of $S$  appearing in
the summands of  $P$ could possibly appear in $P_1$).
In the $v=1$ term it will appear with $\tilde Q $ where $\tilde Q $ is a sum of terms, the first of which is
$(x^{n+1}_{n+1})^f \tilde \Delta^{I_1^1,n+1}_{J_1^1,n+1}  \cdots \tilde  \Delta^{I_f^1,n+1}_{J_f^1,n+1}$ and each appearance will have such a term,
so these add to zero because $\Delta^{I_1^1,n+1}_{J_1^1,n+1}$ in $Mat_{n+1}$, is the same minor as $\Delta^{I_1^1 }_{J_1^1 }$
in $Mat_{n}$. Next is a term say
$(x^{n+1}_{n+1})^{f-1}x^1_{n+1}\tilde \Delta^{1, I_1^1}_{J_1^1,n+1}  \cdots \tilde \Delta^{I_f^1,n+1}_{J_f^1,n+1}$, but then there
must be corresponding terms
$(x^{n+1}_{n+1})^{f-1}x^1_{n+1}\tilde \Delta^{1, I_1^\mu}_{J_1^\mu,n+1}  \cdots \tilde \Delta^{I_f^\mu,n+1}_{J_f^\mu,n+1}$ for
each $2\leq \mu\leq h$. But these must also sum to zero because it is an identity  among minors of the same form
as the original. One continues in this fashion to show all terms in $S$ in the expression of $P_1$ indeed cancel.
\end{proof}

\begin{corollary}\label{hypersurfcor} Fix $k=n-r$ and $S$ with $|S|=k^2-1$,
and allow $n$ to grow.  Then the degrees of the hypersurfaces $\BJ(\s_{n-k,n},L^S)$ grow
at most linearly with respect to $n$.
\end{corollary}

\begin{proof} If  we are in the hypersurface case and $P\in \BI_d(\BJ(\s_{r,n},L^S))$, then even in the worst
possible case where all factors $\Delta^{I^v_s}_{J^v_s}$  in the expression \eqref{pexpr}   but the first have degree one,  the ideal of the hypersurface
$\BJ(\s_{r+u,n+u},L^S)$ is nonempty in degree $(d-r)u$.
\end{proof}

\begin{definition} Let $P$ be a generator of $\BI(\BJ(\s_{r,n},L^S))$ with a presentation of the form \eqref{pexpr}. We say $P$ is
{\it well presented} if $P_q$ constructed as in \eqref{pqexpr}  is a generator of $\BI(\BJ(\s_{r+q,n+q},L^S))$ for all $q$.
\end{definition}

\begin{conjecture}\label{weneedconj} For all $r,n,S$, there exists a set of generators
$P^1\hd P^{\mu}$ of $\BI(\BJ(\s_{r,n},L^S))$ that can be
well presented.
\end{conjecture}

\begin{remark} Well presented expressions are far from unique because of the various Laplace expansions.
\end{remark}

\begin{remark} $\BI(\BJ(\s_{r+q,n+q},L^S))$ may require additional generators beyond the $P^1_q\hd P^{\mu}_q$.
\end{remark}

\section{Examples of equations for $\BJ(\s_r,L^S)$}\label{exsect}

%Many of the equations below were first found for small cases by computer using **Jon fill in here**, and we then
%rewrote and generalized the equations guided by the results of  \S\ref{idealofjs}.
%jml- maybe don't need this remark as it is understood??

\subsection{First examples}\label{firstexamples}
The simplest equations for $\BJ(\s_r,L^S)$ occur when $S$ omits a submatrix of size $r+1$, and one
simply takes the corresponding size $r+1$ minor. The proofs of the following propositions are immediate
consequences of Proposition \ref{usefulprop}, as when one expands each expression, the elements of $S$ cancel.

Consider the example $n=3$, $r=1$, $S=\{x^1_1,x^2_2,x^3_3\}$, and $r=1$.
Then
\begin{equation}
x^2_1x^3_2x^1_3-x^1_2x^2_3x^3_1=M^{23}_{12}x^1_3- M^{12}_{23}x^3_1 \in \BI_3(\BJ(\s_1,L^S)). \label{eq:examplewhichgeneralizes}
\end{equation}
This example generalizes in the following two ways.  
First,   Proposition \ref{degboost} implies: 

%Consider the example $n=3$, $S=\{x^2_2\}$ and $r=1$.
%Then
%$$
%x^1_1x^2_3x^3_2-x^1_2x^2_1x^3_3=x^1_1M^{23}_{23}-x^3_3M^{12}_{12}\in \BI_3(J(\s_1,L^S)).
%$$

\begin{proposition}\label{twoprop} If there are two size $r+1$ submatrices of $Mat_{n}$,  say respectively indexed by
$(I,J)$ and $(K,L)$, that each contain some $x^{i_0}_{j_0}\in S$
but no other point of $S$, then setting $I'=I\backslash i_0$, $J'=J\backslash j_0$, $K'=K\backslash i_0$, $L'=L\backslash j_0$,
  the degree $2r+1$ equations
\be\label{s1}
M^I_JM^{K'}_{L'}-M^{K}_{L}M^{I'}_{J'}
\ene
are in the ideal of $J(\s_r,L^S)$.
\end{proposition}

%Consider the case $n=3$, $r=1$, $S=\{x^1_1,x^2_2,x^3_3\}$  and $P=x^2_1x^3_2x^1_3-
%x^1_2x^2_3x^3_1=M^{23}_{12}x^1_3- M^{12}_{23}x^3_1$.
 By Proposition \ref{degboost},  \eqref{eq:examplewhichgeneralizes} also  generalizes to:

\begin{proposition} \label{tworptwo}   {Suppose that} there exists two  size $r+2$ submatrices of $S$, indexed by $(I,J),(K,L)$,
such that
\begin{enumerate}
\item there are only three elements of $S$ appearing in them, say $x^{i_1}_{j_1}, x^{i_2}_{j_2},x^{i_3}_{j_3}$
with both $i_1,i_2,i_3$ and $j_1,j_2,j_3$ distinct,  {and}
\item each element appears in exactly two of the minors.
\end{enumerate}
Then the degree $2r+1$ equations
\be\label{s3}
M^{I\backslash i_1}_{J\backslash j_1} M^{K\backslash i_2,i_3}_{L\backslash i_2,i_3}
-M^{I\backslash i_2}_{J\backslash j_2} M^{K\backslash i_1,i_3}_{L\backslash i_1,i_3}
\ene
are in the ideal of $J(\s_r,L^S)$.
\end{proposition}

For example, when $S=\{x^1_1,x^2_2,x^3_3\}$,   equation \eqref{s3} may be written
\be\label{eE}
 \Delta^{3}_{2}\Delta^{12}_{13} - \Delta^{2}_{3}\Delta^{13}_{12}.
\ene

 Now  consider the case $n=4$, $r=1$ and $S=\{ x^1_3,x^1_4,x^2_1,x^2_4,x^3_1,x^3_2,x^4_2,x^4_3\}$.
 Proposition \ref{tworptwo} cannot be applied. Instead  we
have  the  equation
$$
x^1_1x^2_2x^3_3x^4_4-x^1_2x^2_3x^3_4x^4_1=M^{12}_{12}x^3_3x^4_4+M^{23}_{13}x^1_2x^4_4+M^{34}_{14}x^2_3x^1_2.
$$
This case generalizes to

\begin{proposition} If there
  are three size $r+1$ submatrices of $Mat_{n\times n}$, indexed by $(I,J),(K,L),(P,Q)$, such that
\begin{enumerate}
\item  the first two contain one element of
$S$ each, say the elements are $x^{i_1}_{j_1}$ for $(I,J)$ and
$x^{i_2}_{j_1}$ for $(K,L)$,
\item   these two elements lie in the same column (or row), and
\item  the third submatrix contains $x^{i_1}_{j_1},x^{i_2}_{j_1}$ and no other element of $S$,
\end{enumerate}
then   the degree $3r+1$ equations
\be\label{s2}
M^I_JM^{K'}_{L'}M^{P'}_{Q'} +M^{K}_{L}M^{I'}_{J'}M^{P'}_{Q''}+ M^{P}_{Q}M^{I'}_{J'}M^{K'}_{L'},
\ene
where $I'=I\backslash i_0$, $J'=J\backslash j_0$, $K'=K\backslash i_0$, $L'=L\backslash j_0$, $P'=P\backslash i_1$, $Q'=Q\backslash j_1$, and $P''=P\backslash i_2$,
are in the ideal of $J(\s_r,L^S)$.
\end{proposition}

 \begin{lemma}\label{Lemma: autoelimination}
  Let $S' \subsetneq S$ with $S' = \{x^1_1 \vvirg x^{n-r}_1\}$ and $S = S' \cup \{ x^{n-r+1}_1 \vvirg x^n_1\}$. Then $\BJ(\sigma_r,L^S) = \BJ(\sigma_r,L^{S'})$.
  \begin{proof}
  Clearly $\BJ(\sigma_r,L^S) \supseteq \BJ(\sigma_r,L^{S'})$. To prove the other inclusion, let $A=(a^i_j) \in \hat \BJ(\sigma_r,L^S)$ be general. Let $\tilde{A}$ be the $r \times r$ submatrix of $A$ given by the last $r$ rows and the last $r$ columns. By generality assumptions, $\tilde{A}$ is non-singular. Therefore, there exist $c_{n-r+1} \vvirg c_n \in \BC$ such that
  \[
   \left(\begin{array}{c}
   a^{n-r+1}_1\\
   \vdots\\
   a^n_1
   \end{array}\right) = \sum_{j={n-r+1}}^n c_j \left(\begin{array}{c}
   a^{n-r+1}_j\\
   \vdots\\
   a^n_j
   \end{array}\right).
  \]
  
  Let $B = (b^i_j)$ be an $n\times n$ matrix such that $b^i_j = a^i_j$ if $j\geq 2$ and $b^i_1 = \sum_{n-r+1}^n c_j b^i_j$. Then $B \in \hat \sigma_r$ and $A \in \hat \BJ([B],L^{S'})
  \subset \hat J(\s_r,L^{S'})$.
  \end{proof}
 \end{lemma}

It will be useful to represent various $S$ pictorially. We will use black diamonds $\blacklozenge$ for entries in $S$ and white diamonds $\lozenge$ for entries omitted by $S$. For example, 
$S=\{ x^1_1,x^2_2\hd x^5_5\}$ is represented by 
$$
\begin{pmatrix}
\blacklozenge& & & & \\
 &\blacklozenge& & &\\
 & &\blacklozenge& & \\
 & & &\blacklozenge&\\ 
 & & & &\blacklozenge
\end{pmatrix}
$$
while $S = \{x^i_j\}_{ij} \smallsetminus \{ x^1_1,x^2_2\hd x^5_5\}$ is represented by
$$
\begin{pmatrix}
\lozenge& & & &\\
&\lozenge& & &\\
& & \lozenge& &\\
& & &\lozenge&\\
& & & &\lozenge
\end{pmatrix}.
$$

\subsection{Case $r=1$}

% If $x^i_j$ is an entry not in $S$ that lies in a row or in a column containing at least $n-r$ %entries of $S$, we say that $x^i_j$ is \emph{implicitly eliminated} by $S$.

\begin{lemma}\label{Lemma: binomial equation}
 Let $S$ be a configuration omitting $x^1_1 \vvirg x^k_k,x^1_2 \vvirg x^{k-1}_k$
 and $x^k_1$ for some $k \geq 2$. Then $\BI_k(\BJ(\sigma_1,L^S))$ contains the binomial
 \[
  x^1_1 \cdots x^k_k - x^1_2 \cdots x^k_1
 \]
 Moreover, if the complement $S^c = \{x^1_1 \vvirg x^k_k,x^1_2 \vvirg x^{k-1}_k,x^k_1\}$, then $\BJ(\sigma_1,L^S)$ is a hypersurface.
 \end{lemma}
\begin{proof}
Let $ f = x^1_1 \cdots x^k_k - x^1_2 \cdots x^k_1$. It suffices to show that $f\in \BI(\sigma_1)$, namely that it can be generated by $2 \times 2$ minors. If $k = 2$, then $f$ is the $2 \times 2$ minor $M^{12}_{12}$. Suppose $k \geq 3$. 

Define $f_2 = x^1_1 x^2_2 - x^1_2 x^2_1 = M^{12}_{12}$. For any $j = 3 \vvirg k$ define 
\[ f_j = x^j_j f_{j-1} - x^1_2 \cdots x^{j-2}_{j-1}  M^{j-1,j}_{1,j} . \]
Thus  $f_j = x^1_1 \cdots x^j_j - x^1_2 \cdots x^j_1\in \BI(\sigma_1)$ for all  $j=3\hd k$   and  $f_k = f$.

The last assertion  follows because 
for  any $S' \supsetneq S$,       iterated applications of Lemma \ref{Lemma: autoelimination} implies $\BJ(\sigma_1,L^{S'}) = Mat_{n \times n}$.
\end{proof}

\begin{lemma}\label{fulfirlem}
 Let $S$ be a configuration    omitting at least two entries in each row and in each column. Then there exists $k \geq 2$ such that, up to a permutation of rows and columns, 
 $S^c\supseteq \{x^1_1 \vvirg x^k_k,x^1_2 \vvirg x^k_1\}$.
 \end{lemma}

\begin{proof}
 After a  permutation, we may assume $x^1_1 \in S^c$, and, since $S$ omits at least another entry in the first column,   $x^2_1 \in S^c$. Since $S$ omits at least $2$ entries in the second row, assume $x^2_2 \in S^c$. $S$ omits at least another entry in the second column: if that  entry is $x^1_2$, then $k=2$ and $S$ omits a $2\times 2$ minor; otherwise we may assume $x^3_2 \in S^c$. Again $S$ omits another entry on the third row: if that  entry is $x^3_1$ (resp.  $x^3_2$), then   $k=3$ (resp.  $k=2$) and $S$ omits a set of the desired form. After at most $2n$ steps, this procedure terminates, giving a $k\times k$ submatrix $K$ with one of the following configurations, one the tranpose of the other:
 \[
 \begingroup
\arraycolsep=6pt\def\arraystretch{1}
 \left[\begin{array}{cccc}
\lozenge & \lozenge & & \\
&  \ddots &  \ddots & \\
 &  & \ddots & \lozenge \\
\lozenge & & & \lozenge \\             
            \end{array}\right],
            \qquad
 \left[\begin{array}{cccc}
\lozenge &  & & \lozenge \\
\lozenge &  \ddots &   & \\
 &  \ddots & \ddots & \\
 & & \lozenge & \lozenge \\             
            \end{array}\right].
\endgroup
 \]
  $K$ and its transpose are equivalent under permutations of rows and columns
  because  $K^T = PKP$ where $P$ is the $k \times k$ permutation matrix having $1$ on the anti-diagonal and $0$ elsewhere.
\end{proof}

 \begin{lemma}\label{Lemma: two entries in each row and each column}
 Let $S$ be a configuration of $n^2-2n$ entries. Then there exist $k \in [n]$ and a $k\times k$ submatrix $K$ such that, up to a permutation of rows and columns, at least $2k$ entries of the complement $S^c$ of $S$ lie in $K$ in the following configuration
 \begin{equation}\label{eqn: Lemma: two entries in each row and in each column}
 \left(
   \begin{array}{cccc|cccc}
\lozenge & \lozenge & & & & & &\\
&  \rotatebox{-10}{$\ddots$} &  \rotatebox{-10}{$\ddots$} & & & & &\\
 &  & \rotatebox{-10}{$\ddots$} & \lozenge & & & &\\
\lozenge & & & \lozenge & & & & \\
\hline
& & & & & & &  \\
& & & & & & &  \\
& & & & & & &  \\
 \end{array}\right).
 \end{equation}
  Moreover, if $\BJ(\sigma_1,L^S)$ is a hypersurface  then these are the only omitted entries in $K$ and the ideal of $\BJ(\sigma_1,L^S)$ is 
  generated by $\left(  x^1_1 \cdots x^k_k - x^1_2 \cdots x^k_1\right)$. 
  \end{lemma}
\begin{proof}
  To prove the first assertion, we proceed by induction on $n$. The case $n=2$ provides $s=0$ and $k =2$ trivially satisfies the statement.

If $S$ omits at least (and therefore exactly) $2$ entries in each row and in each column, then  we conclude by Lemma \ref {fulfirlem}.

Suppose that $S$ contains an entire row (or an entire column). Then $S^c$ is concentrated in a $(n-1) \times n$ submatrix. In this case we may consider an $(n-1)\times(n-1)$ submatrix obtained by removing a column that contains at most $2$ entries of $S^c$. Thus, up to reduction to a smaller matrix, we may always assume that $S$ omits at least one entry in every row (or at least one entry in every column).

After a  permutation, we may assume that the first row omits $x^1_1$. If the first column omits at most one more entry, then $S$ omits at least $2(n-1)$ entries in the the submatrix obtained by removing the first row and the first column. We conclude by induction that there exists a $k$ and a $k\times k$ submatrix with the desired configuration in the submatrix.

If the first column omits at least $2$ entries other than $x^1_1$, then there is another column omitting only $1$ entry. Consider the submatrix obtained by removing this column and the first row: $S$ omits exactly $2(n-1)$ entries in this submatrix, and again we conclude by induction.

To prove the last assertion,  if other omitted entries lie in $K$, then they provide another equation for $\BJ(\sigma_1,L^S)$.
\end{proof}

\begin{theorem}\label{thm: hypersurface for r=1}
The number of irreducible components of $\cR[n,1,n^2-2n]$ coincides with the number of cycles of the complete bipartite graph $K_{n,n}$. Moreover, every ideal of an irreducible component is generated by a binomial of the form
\[
 x^{i_1}_{j_1} \cdots x^{i_k}_{j_k} -  x^{i_1}_{j_{\tau(1)}} \cdots x^{i_k}_{j_{\tau(k)}}, 
\]
for some $k$, where $\tau \in \frak S_k$ is a cycle and $I,J\subset[n]$ have
size $k$.
\end{theorem}
\begin{proof}
$\cR[n,1,n^2-2n]$ is equidimensional and its irreducible components are $\BJ(\sigma_1,L^S)$ where $S$ is a configuration of entries providing a join of expected dimension.

By Lemma \ref{Lemma: two entries in each row and each column} there exists a $k$ such that, up to a permutation of rows and columns, $S$ omits the entries $ x^1_1 \vvirg x^k_k, x^1_2 \vvirg x^k_1$ and the equation of $\BJ(\sigma_1,L^S)$ is $ x^1_1 \cdots x^k_k - x^1_2 \cdots x^k_1 = 0$. In particular, entries in $S^c$ that do not lie in the submatrix $K$ are free to vary. Let $F$ be the set of entries whose complement is  $  \{x^1_1 \vvirg x^k_k, x^1_2 \vvirg x^k_1 \}$; we obtain $\BJ(\sigma_1,L^S) = \BJ(\sigma_1,L^F)$. This shows that the irreducible components are determined by the choice of a $k\times k$ submatrix and by the choice, in this submatrix, of a configuration of $2k$ entries such that, after a permutation of rows and columns, it has the form of \eqref{eqn: Lemma: two entries in each row and in each column}.

Every configuration of this type, viewed as the adjacency matrix of a $(n,n)$-bipartite graph, determines a cycle in the complete bipartite graph $K_{n,n}$. This shows that the number of irreducible components of $\cR[n,1,n^2-2n]$ is the number of such cycles\end{proof}
 
\begin{remark} More precisely, the number of irreducible hypersurfaces of degree
$k$ in $\cR[n,1,n^2-2n]$ coincides with the number of cycles in $K_{n,n}$
of length $2k$. In particular, for every $k$ with  $2\leq k\leq n$, $\cR[n,1,n^2-2n]$ has
exactly $\binom{n}{k}^2 \frac{k! (k-1)!}{2}$ irreducible components of
degree $k$. The total number of irreducible components is
\[
 \sum_{k=2}^n  \binom{n}{k}^2 \frac{k! (k-1)!}{2}.
\]
\end{remark}

%\begin{remark} The number of cycles in $K_{n,n}$ is 
%\[
% \sum_{k=2}^n  \binom{n}{k}^2 \frac{k! (k-1)!}{2}.
%\]
%\end{remark}

\begin{example} Examples of generators of ideals of  $\BJ(\s_1,L^S)$:
\begin{enumerate}
\item\label{cs1} If $s=1$, the ideal is generated by the $2\times 2$ minors
not including the element of $S$
and the degree is $\tdeg(\s_1)-1=\frac{(2n-2)!}{[(n-1)!]^2}-1$.
\item\label{cs2} If $s=2$, the ideal is generated by the $2\times 2$ minors
not including the elements of $S$.
If the elements of $S$ lie in the same column or row, the degree is $
\tdeg(\s_1)-n=\frac{(2n-2)!}{[(n-1)!]^2}-n$ and otherwise it is
$ \tdeg(\s_1)-2=\frac{(2n-2)!}{[(n-1)!]^2}-2$.
\item If $s=3$ and there are no entries of $S$ in the same row or column,
the ideal is generated in
degrees two and three by the $2\times 2$ minors not including the elements
of $S$ and the difference
of the two terms   in the $3\times 3$ minor containing all three elements
of $S$
and the degree is $\frac{(2n-2)!}{[(n-1)!]^2}-3$.
\item If $s=3$ and there are two entries of $S$ in the same row or column,
the ideal is generated in
degree  two   by the $2\times 2$ minors not including the elements of $S$.
%and the degree is $\frac{(2n-2)!}{[(n-1)!]^2}-2n+3$  **check** in the case
of the three entries all in the same.
%row or column and $\frac{(2n-2)!}{[(n-1)!]^2}-[n-1+\binom n2]$ otherwise.
\item
If $s\leq n$ and there are no elements of $S$ in the same row or column,
then $\tdeg(\BJ(\s_1,L^S)=\frac{(2n-2)!}{[(n-1)!]^2}-s$. (See Theorem
\ref{degreethm}.)
\end{enumerate}
\end{example}

\begin{proof}
 Parts (1),(3),(5), as well as part (2) when the entries do not lie in the
same column of row, are consequence of Theorem \ref{degreethm}. Moreover
the generators of the ideal can be obtained from Proposition \ref{jjens}.
For  part (2), we prove that $\deg(\BJ(\s_1,L^S)) = \deg \sigma_1 - n$ when
the two entries of $S$ lie in the same row or in the same column. Assume $S
= \{x^1_1, x^1_2 \}$, so that $\BJ(\sigma_1, L^S) =
\BJ(\BJ(\sigma_1,  [a^1\otimes b_1]  ), [a^1\otimes b_2] )$ for some basis
vectors $a^1$ and $b_1,b_2$. From (1), equations for
$\BJ(\sigma_1, [a^1\otimes b_1] )$ are $2\times 2$ minors not involving
the variable $x^1_1$, and $\deg \BJ(\sigma_1, [a^1\otimes b_1] ) = \deg
\sigma_1 -1$.

 The ideal of the tangent cone $TC_{[a^1\otimes b_2]}
\BJ(\sigma_1, [a^1\otimes b_1] )$ is generated by the variables $x^j_k$
for $j,k \geq 2$ (obtained as the lowest degree term in the coefficient
$(x^1_2-1)^0$ in the expansion $ (x^1_2-1)^0 (x^j_k - x^1_k x^j_2) +
(x^1_2-1)x^j_k$ of $x^1_2 x^j_k - x^1_k x^j_2$) and by the minors $x^i_1
x^j_2 - x^i_2 x^j_1$, with $i,j \geq 1$. Therefore $TC_{[a^1\otimes b_2]}
\BJ(\sigma_1, [a^1\otimes b_1] )$ has the same degree as the variety of
matrices of size $(n-1) \times 2$ and rank at most $1$, that is $n-1$. From
Proposition \ref{conedegprop}, we conclude.
\end{proof}

\subsection{Case $r=2$}
The following propositions are straight-forward to verify with the help of a computer with explicit computations available at \url{www.nd.edu/~jhauenst/rigidity}. 

\begin{proposition}
Let $n=5$, $r=2$ and let  $S=\{ x^1_1,x^2_2\hd x^5_5\}$. Then $\BJ(\s_{2,5},L^S)$
has $27$ generators of degree $5$ of the form \eqref{s1}, e.g.,
$M^{123}_{456}M^{45}_{12}-M^{345}_{123}M^{12}_{45}$, and $5$ generators  of degree $6$   with   $6$ summands in their expression, each of the form $M^{\bullet\bullet\bullet}_{\bullet\bullet\bullet}M^\bullet_\bullet M^\bullet_\bullet M^\bullet_\bullet$:
\begin{align*}
&- M^{345}_{123}M^{1}_{4}M^{1}_{5}M^{2}_{1} +
M^{235}_{134}M^{1}_{2}M^{1}_{5}M^{4}_{1} -
M^{234}_{135}M^{1}_{2}M^{1}_{4}M^{5}_{1}
\\
& +
M^{134}_{235}M^{1}_{4}M^{2}_{1}M^{5}_{1} -
M^{123}_{345}M^{1}_{2}M^{4}_{1}M^{5}_{1} -
M^{135}_{234}M^{1}_{5}M^{2}_{1}M^{4}_{1}.
\end{align*}
\end{proposition}

\begin{proposition}
Let $n=6$, $r=2$, $s=15$ and let  $S$ be given by
$$\begin{pmatrix}
0&0&0&0&\blacklozenge&\blacklozenge\\
\blacklozenge&0&0&\blacklozenge&0&0\\
0&\blacklozenge&\blacklozenge&0&0&0\\
0&\blacklozenge&\blacklozenge&0&\blacklozenge&0\\
\blacklozenge&0&0&\blacklozenge&0&\blacklozenge\\
0&0&0&\blacklozenge&\blacklozenge&\blacklozenge
\end{pmatrix}.
$$
Then
$\BJ(\s_{2,6},L^S)$ is a hypersurface of degree  {9}   whose equation  is:
\begin{align}
&\label{s5} - M^{235}_{235}M^{12}_{36}M^{16}_{12}M^{34}_{14} +
M^{235}_{235}M^{12}_{26}M^{16}_{13}M^{34}_{14} +
M^{126}_{236}M^{13}_{13}M^{25}_{25}M^{34}_{14}
 -
M^{126}_{236}M^{13}_{12}M^{25}_{35}M^{34}_{14}
\\
\nonumber &
 +
M^{126}_{235}M^{13}_{16}M^{25}_{23}M^{34}_{14}
 -
M^{126}_{235}M^{13}_{14}M^{25}_{23}M^{34}_{16} +
M^{134}_{146}M^{12}_{23}M^{25}_{23}M^{36}_{15} -
M^{134}_{146}M^{13}_{15}M^{25}_{23}M^{26}_{23}
\\
\nonumber & -
M^{136}_{136}M^{12}_{23}M^{25}_{25}M^{34}_{14} +
M^{136}_{126}M^{12}_{23}M^{25}_{35}M^{34}_{14}.
\end{align}
\end{proposition}

The weight of equation \eqref{s5} is $(1^2,2^2,3^2,4,5,6)\times(1^2,2^2,3^2,4,5,6)$.
(This weight is hinted at because the first, second and third columns and rows each have two elements
of $S$ in them and the fourth, fifth and sixth rows and columns each have three.)

\begin{proposition} Let $n=6$, $r=2$, $s=15$,  and let  $S$ be given by
\label{n6r2s15}
%$$
%\begin{pmatrix}
%0&0&0&0&X&X\\
%X&0&0&X&0&0 \\
%0&X&X&0&0&0\\
%0&X&X&0&X&0 \\
%X&0&0&X&0&X \\
%0&0&0&X&X&X
%\end{pmatrix}
%$$
$$
\begin{pmatrix}
0&0&0&0&\blacklozenge&\blacklozenge \\
0&\blacklozenge&0&0&\blacklozenge&0  \\
0&0&\blacklozenge&\blacklozenge&0&0  \\
\blacklozenge&0&0&\blacklozenge&\blacklozenge&0  \\
\blacklozenge&0&\blacklozenge&0&0&\blacklozenge  \\
0&\blacklozenge&0&\blacklozenge&0&\blacklozenge
\end{pmatrix}.
$$
Then,
$\BJ(\s_{2,6},L^S)$ is a hypersurface of   degree $16$.
\end{proposition}

We do not have a concise expression for the equation of $\BJ(\s_{2,6},L^S)$. Expressed na\"\i vely, it is the sum of $96$ monomials, each
with coefficient $\pm 1$ plus two monomials with coefficient $\pm 2$, for a total of $100$ monomials counted
with multiplicity. The monomials are of weight $(1^4,2^3,3^3,4^2,5^2,6^2)\times (1^4,2^3,3^3,4^2,5^2,6^2)$.
(This weight is hinted at because the first, second and third columns and rows each have two elements
of $S$ in them, but the first is different because in the second and third a column element equals
a row element,  and the fourth, fifth and sixth rows and columns each have three.)

\subsection{Case $r=n-2$}

Proposition \ref{tworptwo} implies:

\begin{theorem}\label{rnmtwothm} In the hypersurface case $r=n-2$, $s=3$,
there are two types of varieties up to isomorphism:
\begin{enumerate}
\item If no two elements of $S$ are in the same row or column, then
the hypersurface is of degree $2n-3$ and can be represented by an equation of the form \eqref{eE}.
There are $6{\binom n3}^2$ such components, and they are of  of maximal degree.

\item
If two elements are in the same row and one in a different column  from those two,
or such that
one element shares a row with one and a column with the other,
then the equation
is the unique size $(n-1)$ minor that has no elements of $S$ in it.
There are $n^2$ such components.
\end{enumerate}

If all three elements of $S$ lie on a row or column, then $\BJ(\s_{n-2},L^S)$ is not a hypersurface.
\end{theorem}

\begin{corollary}\label{corollary: n-2-maximal rigidity}  Let $M$ be an $n\times n$  matrix. Then $M$ is
maximally $(n-2)$-border rigid
if and only if  no size $n-1$ minor is zero and
  for all   
$(i_1,i_2,i_3)$ taken  from distinct elements of $[n]$, and  all   
$(j_1,j_2,j_3)$ taken  from distinct elements of $[n]$, the equation
$\Delta^{i_1}_{j_1}\Delta^{i_2i_3}_{j_2j_3} - \Delta^{i_2}_{j_2}\Delta^{i_1i_3}_{j_1j_3}$
does not vanish on $M$. 
\end{corollary}
\section{Varieties of matrices}\label{matvarsect}

\subsection{General remarks}

 Recall the construction of matrices from directed acyclic graphs in \S \ref{sec:intro}. 
To each graph $\G$ that is the disjoint union of directed acyclic graphs with $n$ input gates  and $n$ output
gates we associate the set $\Sigma_{\G}^0\subset Mat_n$ of all matrices admitting a linear circuit  (see~\S\ref{sec:intro})  with
underlying graph $\G$. We let $\Sigma_{\G}:=\ol{\Sigma_{\G}^0}\subset Mat_n$, the {\it variety of linear
circuits~associated~to~$\G$}.

For example $\cR[n,r,s]^0=\cup \Sigma_{\G}^0$ where the union is over all  $\G=\G_1 + \G_2$ (addition as in Figure~\ref{fig:sumofgraphs})  where
$\G_1$ is of depth two  with $r$ vertices at the second level and  is  a complete bipartite graph at each level, and $\G_2$ is of depth one, with
$s$ edges.

\begin{proposition}\label{bigvarprop} Let $\Sigma \subset Mat_n$ be a variety of dimension $\d$. Then a general
element of $\Sigma$ cannot be computed by a circuit of size $\d-1$.
\end{proposition}
\begin{proof} Let $\G$ be a fixed graph representing a family of linear circuits with $\g$ edges. Then $\G$
can be used for at most a $\g$-dimensional family of matrices. Any variety of matrices of dimension greater
than $\g$ cannot be represented by $\G$, and since there are a finite number of graphs of size at most $\g$,
the dimension of their union is still $\g$.
\end{proof}

\subsection{Cauchy matrices}
Let $1\leq i,j,\leq n$. Consider the rational map
\begin{align}
\label{cauchymap} Cau_n: \BC^n\times \BC^n &\dashrightarrow Mat_n\\
\nonumber ((x^i),(z_j))&\mapsto (y^i_j):=\frac 1{x^i+z_j}
\end{align}
The variety of Cauchy matrices $Cauchy_n\subset Mat_n$ is defined to be the closure of the  image of~\eqref{cauchymap}.
It  is   $\FS_n\times \FS_n$ invariant and has dimension $2n-1$. To see the
dimension,  note that
$Cauchy_n$   is the {\it Hadamard inverse}  or {\it Cremona transform} of a linear subspace of $Mat_{n}$ of dimension $2n-1$
(that is contained in $\s_2$).
The Cremona map is
\begin{align*}
Crem_N: \BC^N&\dashrightarrow \BC^N\\
(w_1\hd w_N)&\mapsto \left(\frac 1{w_1}\hd \frac 1{w_N}\right)
\end{align*}
which is generically one to one.
The fiber
of $Crem_{n^2}\circ Cau_n$  over $(x^i+z_j)$ is
$((x^i+\l),(z_j-\l))$, with $\l\in \BC$.

One can obtain equations for $Cauchy_n$ by transporting the linear equations of its Cremona transform, which are the
  $(n-1)^2$ linear  equations, e.g.,
for $ i,j = 2\hd n$,  $y^1_1 + y^i_j - y^i_1 - y^1_j$.
(More generally,  it satisfies the equation $y^{i_1}_{j_1} + y^{i_1}_{j_2} - y^{i_2}_{j_1} - y^{i_1}_{j_2}$  for all $i_1,j_1,i_2,j_2$.)
Thus, taking reciprocals and clearing denominators,  the Cauchy variety has cubic equations
$$
    y^{i_1}_{j_2}   y^{i_2}_{j_1}  y^{i_1}_{j_2}
+y^{i_1}_{j_1}     y^{i_2}_{j_1}  y^{i_1}_{j_2}
-y^{i_1}_{j_1}   y^{i_1}_{j_2}     y^{i_1}_{j_2}
-y^{i_1}_{j_1}   y^{i_1}_{j_2}   y^{i_2}_{j_1}.
$$

Alternatively, $Cauchy_n$ can be parametrized by the first row and column: let $ 2  \leq \rho,\s\leq n$, and
denote the entries of $A$ by $a^i_j$. Then the space is parametrized by
$a^1_1,a^{\rho}_1,a^1_{\s}$, by setting  $ a^{\rho}_{\s} =[\frac 1{a^{\rho}_1}
+\frac 1{a^1_{\s}}-\frac 1{a^1_1}]^{-1}$.

Any square submatrix of a Cauchy matrix is a Cauchy matrix, and the determinant of a Cauchy matrix is given by
\be\label{cauchydet}
\frac{\prod_{i<j}(x^i-x^j)\prod_{i<j}(z^i-z^j)}{\prod_{i,j}(x^i+z_j)}
\ene
In particular, if $x^i, -z_j$ are all distinct, then  all minors of the Cauchy matrix are nonzero.

\begin{proposition}\label{cauchyrig} A general Cauchy matrix is both maximally $r=1$ rigid and maximally $r=n-2$ rigid.
\end{proposition}
\begin{proof}
For the $r=1$ case,  let  $\s$ be  a $k$-cycle in $\mathfrak{S}_k$ and  say there were an equation
$$
y^{i_1}_{j_1}\cdots y^{i_k}_{j_k} - y^{i_1}_{j_{\s(1)}}\cdots y^{i_k}_{j_{\s(k)}}
$$
 $Cauchy_n$  satisfied. By the $\FS_n\times \FS_n$ invariance we may assume the equation is
$$
y^{ 1}_{ 1}\cdots y^{ k}_{ k} - y^{ 1}_{ \s(1) }\cdots y^{ k}_{ \s(k) }
$$
which may be rewritten as
$$
 \frac 1{y^{ 1}_{ 1}\cdots y^{ k}_{ k} }=\frac 1{y^{ 1}_{ \s(1) }\cdots y^{ k}_{ \s(k) }}
$$
The first term contains the monomial $x^1 \cdots x^{k-1} z_k$, but the second does not.

For the $r=n-2$ case, we may assume the equation is
$\Delta^1_2\Delta^{23}_{13}-\Delta^2_1\Delta^{13}_{23}$,
because for a general Cauchy matrix all size $n-2$ minors are nonzero and
we have the $\FS_n\times \FS_n$ invariance.

Write our equation as $\frac{\Delta^1_2}{\Delta^2_1}= \frac{\Delta^{13}_{23}}{\Delta^{23}_{13}}$.
Then using \eqref{cauchydet} and canceling all repeated terms we get
$$
\frac{(x^2-x^3)(z_1-z_3)(x^1+z_3)(x^3+z_2)}{(x^1-x^3)(z_2-z_3)(x^2+z_3)(x^3+z_1)}=1
$$
which fails to hold for a general Cauchy matrix.
\end{proof}

\subsection{The Vandermonde variety} 
In \cite[ p26]{MR2870721}, they ask if a general Vandermonde matrix has maximal rigidity.

Consider the   map
\begin{align}\label{vandmap} \Van_n: \BC^{n+1}&\ra \Mat_n\\
\nonumber (y_0,y_1\hd y_n)&\mapsto
\begin{pmatrix}
(y_0)^{n-1} & \cdots & (y_0)^{n-1}\\
(y_0)^{n-2} y_1 & \cdots & (y_0)^{n-2}y_n\\
(y_0)^{n-3} y_1^2 & \cdots & (y_0)^{n-3} (y_n)^2 \\
& \vdots &\\
(y_1)^{n-1} &\cdots & (y_n)^{n-1}
\end{pmatrix}=(x^i_j)
\end{align}
Define the {\it Vandermonde variety} $\Vand_n$ to be the closure of the image of this map. Note that this variety contains
 $n$ rational normal curves (set all $y_j$ except $y_0,y_{i_0}$ to zero), and is $\FS_n$-invariant (permutation of columns).
The (un-normalized) Vandermonde matrices are the Zariski open subset where $y_0\neq 0$ (set $y_0=1$ to obtain the
usual Vandermonde matrices).
Give $Mat_{n\times n}$ coordinates $x^i_j$.
The variety $\Vand_n$ is contained in the linear space $\{ x^1_1-x^1_2=0\hd x^1_1-x^1_n=0\}$ and it is the zero set of  these
linear equations and the generators of the ideals of the rational normal curves $\Van[y_0,0\hd 0,y_j,0\hd 0]$. Explicitly, fix $j$, the generators
for the rational normal curves are
the $2\times2$ minors of
$$
\begin{pmatrix} x^1_j & x^2_j & \cdots &x^{n-1}_j\\
x^2_j & x^3_j & \cdots &x^{n}_j
\end{pmatrix} 
$$
see, e.g., \cite[p. 14]{Harris}, and thus the equations for the variety
are, fixing $j$ and $i<k$,    the quadratic polynomials $x^i_jx^k_j-x^{i+1}_jx^{k-1}_j$.

To see the assertion about the zero set, first consider the larger parametrized
variety where instead of $y_0$ appearing in each column, in the $j$-th column,
replace $y_0$ by a variable $y_{0j}$. The resulting variety is the join
of $n$ rational normal curves, each contained in a $\pp{n-1}\subset \BP Mat_n$,
where the $\pp{n-1}$'s are just the various columns. 
In general, given varieties $X_j\subset \BP V_j$, $j=1\hd q$, the
join $J(X_1\hd X_q)\subset \BP (V_1\opc V_q)$ has ideal generated
by $I(X_1,\BP V_1)\hd I(X_q,\BP V_q)$, see, e.g. \cite[18.17, \emph{Calculation III}]{Harris}.
The second set of equations exactly describes this join.  Now intersect this
variety with the linear space where  all entries on the first row are set equal.
We obtain the Vandermonde variety.

\begin{proposition}\label{vandrigid}
  $\Vand_n\not\subset \cR[{n,1,n^2-2n}]$ and   $\Vand_n\not\subset \cR[{n,n-2,3}]$, i.e.,   Vandermonde matrices are generically maximally $1$-border
 rigid
and  $(n-2)$-border rigid.
\end{proposition}
\begin{proof}
Say we had $Vand_n$ in some component of $\cR[{n,1,n^2-2n}]$. 
Using the $\FS_n$-invariance, we may assume the equation it satisfies  is
$x^{i_1}_1\cdots x^{i_k}_k-x^{i_1}_{\s(1)}\cdots x^{i_k}_{\s(k)}$ for some $k$, where $\sigma \in \FS_k$ is a $k$-cycle. Assume $i_k = \max\{i_\ell\}$. Then the first monomial is divisible by $(y_k)^{i_k-1}$ but the second is not.

For the $n-2$-rigidity, since no minors are zero, by  Corollary \ref{corollary: n-2-maximal rigidity}
and the $\FS_n$-invariance, it suffices
to consider equations of the form
$\Delta^j_2\Delta^{ik}_{13}-\Delta^k_3\Delta^{ij}_{12}$,
where $S = \{x^i_1,x^j_2,x^k_3\}$.
First consider the case that $2 \notin \{i,j,k\}$.
The $y_2$-linear coefficient of $\Delta^j_2\Delta^{ik}_{13}-\Delta^k_3\Delta^{ij}_{12}$
is
$\Delta^j_2\Delta^{ik2}_{132}-\Delta^{k2}_{32}\Delta^{ij}_{12}$.
This expression is nonzero, because as a polynomial in $y_1$ it has linear coefficient
$\Delta^{j2}_{21}\Delta^{ik2}_{132}$, which is a product of minors and hence nonzero.
Now for the other case let $i=2$, $j \neq 2$, $k \neq 2$.
But the $y_2$-linear coefficient of $\Delta^j_2\Delta^{2k}_{13}-\Delta^k_3\Delta^{2j}_{12}$
is $-\Delta^{k2}_{32}\Delta^{2j}_{12}$, which is also nonzero.
\end{proof}

\subsection{The DFT matrix}
The following \lq\lq folklore result\rq\rq\ was communicated to us (independently)  by A. Kumar and A. Wigderson:
\begin{proposition} Let $A$ be a matrix with an eigenvalue of multiplicity $k>\sqrt n$. Then $A\in \hat\cR[n,n-k,n]^0$.
\end{proposition}
\begin{proof} Let $\l$ be the eigenvalue with multiplicity $k$, then $A-\l Id$ has rank $n-k$.
To have the condition be nontrivial, we need $r(2n-r)+s=(n-k)(2n-(n-k))+n<n^2$, i.e., $n<k^2$.  
\end{proof}

Equations for the variety of matrices with eigenvalues of high multiplicity can be
obtained via resultants applied to the coefficients of the characteristic polynomial of a matrix.

\begin{corollary}Let $n=2^k$, then  $DFT_{n}\in \hat\cR[n,\frac{3n}4,n]^0$.
\end{corollary}
\begin{proof} The eigenvalues of $DFT_n$ are $\pm 1,\pm \sqrt{-1}$ with multiplicity roughly $\frac n4$ each.
\end{proof}

 \begin{proposition} Any matrix with $\BZ_2$ symmetry (either symmetric or symmetric about the anti-diagonal) is not
maximally $1$-border rigid.
 \end{proposition}
 \begin{proof}
 Say the matrix is symmetric. Then $x^1_2x^2_3\cdots x^{n-1}_nx^n_1 - x^2_1 x^3_2\cdots x^{n}_{n-1}x^1_n$ is in the ideal
 of the hypersurface $\BJ(\s_1,L^S)$ where  $S$  is the span of all the entries not appearing in the expression.
 \end{proof}
% \begin{proof}
% Say the matrix is symmetric. Let $\sigma \in \FS_{n}$ be the cycle defined by $\sigma: i \mapsto i-2$ (modulo $n$). This is an $n$-cycle if $n$ and an $\frac{n}{2}$-cycle if $n$ is even. If $n$ is odd, $x^1_2 x^2_3 \vvirg x^{n-1}_n x^n_1 - x^1_{\sigma(2)} \cdots x^n_{\sigma(1)}$ is in the ideal of the hypersurface $\BJ(\s_1,L^S)$ where  $S$  is the span of all the entries not appearing in the expression.
% \end{proof}

\subsection{The DFT curve}
We define two  varieties that contain the DFT matrix, the first corresponds to a curve in projective space.

Define the {\it DFT curve} $CDFT_n\in Mat_n$ to be the image
of the map
\begin{align}
\label{dftcurve} \BC^2&\ra Mat_n\\
\nonumber (x,w)&\mapsto
\begin{pmatrix}
x^{n-1}& x^{n-1}& x^{n-1}& \cdots &x^{n-1}\\
x^{n-1}& x^{n-2}w&x^{n-3}w^2& \cdots & w^{n-1}\\
&\vdots & & &\\
x^{n-1}&  w^{n-1}&x^{1}w^{n-2}& \cdots & x^{n-2}w
\end{pmatrix}
\end{align}

This curve is a subvariety of $\Vand_n$ where $y_0=y_1=x$ and $y_j=w^{j-1}$. From this one obtains its equations.

\begin{proposition}   For general $w$, and even 
 for $w$ a fifth root of unity, the matrix
$$
M(w):=\begin{pmatrix} 1&1&1&1&1\\ 1&w&w^2&w^3&w^4\\ 1&w^2&w^4&w&w^3\\ 1&w^3&w&w^4&w^2\\ 1&w^4&w^3&w^2&w
\end{pmatrix}
$$
satisfies $M(w)\not\in \hat\cR[5,3,2]$, 
$M(w)\in\cR[5,3,3]^0$, $M(w)\not\in \hat\cR[5,1,12]$,
and $M(w) \in \hat\cR[5,1,13]^0$.
\end{proposition}

This is proved by explicit calculation at \url{www.nd.edu/~jhauenst/rigidity}. 
For a more general DFT matrix we have: 

\begin{proposition}
Let $p$ be prime, then the DFT curve $CDFT_p$  satisfies,
or all $A\in CDFT_p$
$$\rig_1(A)\leq (p-1)^2+1-(p-1).$$  
In other words, $CDFT_p\subset \hat\cR[p,1,p^2-3p+3]^0$.
\end{proposition}
\begin{proof}
Change the $(1,1)$-entry to $w\inv$ and all entries in the lower right $(p-1)\times (p-1)$ submatrix not already
equal to $w$ to $w$. The resulting matrix is
$$
\begin{pmatrix}
w\inv & 1 & 1 & \cdots & 1\\
1 & w& w&\cdots &w\\
& & \vdots & & \\
1 & w& w&\cdots &w
\end{pmatrix}
$$
\end{proof}

\subsection{The variety of factorizable matrices/the butterfly variety}
 The DFT algorithm may   be thought of as factorizing the size $n=2^k$    DFT matrix into a product
of $k$ matrices $S_1\hd  S_k$ with each $S_i$ having $2n$ nonzero entries.

If $S_j$, $1\leq j\leq d$  are matrices with $s_j$ nonzero entries in $S_j$,  with $s_j=f_jn$ 
for some natural numbers $f_j$, then $S_1S_2$ has at most $f_1f_2n$ nonzero
entries. Consider the set of matrices $A$ such that we may write
$A=S_1\cdots S_d$ with $s_j=f_jn$ and $f_1\cdots f_d=n$. Then $A$ may be computed
by a linear circuit of depth $d$ and size $(f_1+\cdots +f_d)n$.  In the DFT case we have $f_j=2$ and $d=\tlog(n)$.

This space of matrices is the union of a large number of components, each component is the image of a map:
$$
\Bfly: \hat L_{s_1}\ctimes \hat L_{s_d} \ra Mat_{n\times n}
$$
where $\hat L_{s_j}\subset Mat_{n\times n}$ is the span of some $S\subset \{ x^i_j\}$ of cardinality $s_j$.
In the most efficient configurations (those where the map has the smallest dimensional  fibers), each entry $y^i_j$ in a matrix in the image will be of the form
$y^i_j=(x_1)^i_{j_1}(x_2)^{j_1}_{j_2}\cdots (x_d)^{j_{d-1}}_{j}
$
where the $j_u$'s are fixed indices (no sum).

If we are not optimally efficient, then the equations for the corresponding variety become more complicated, and the dimension will
drop.

From now on, for simplicity  assume $n=2^k$, $d=k$ and $s_j=2n$ for $1\leq j\leq k$.
Let $FM_{n}^0
$ denote the set  of factorizable or {\it butterfly}  matrices, 
the set of matrices
$A$ such that $A=S_1\cdots S_k$ with $S_j$ as above, and let 
  $FM_n:=\ol{FM_n^0}$ denote its Zariski closure. The term \lq\lq butterfly\rq\rq\ comes from the name commonly used for the corresponding circuit, e.g., see \cite[\S 3.7]{MR1137272}.
By construction every $A\in FM_n^0$ admits a linear circuit of size $2n\tlog n$, see,  e.g.,  Figure~\ref{fig:butterfly}: 
 {t}he graph has $48$ edges compared with $64$ for a generic $8\times 8$ matrix, and in general one has
$2^{k+1}k=2n\tlog~n$ edges compared with $2^{2k}=n^2$ for a generic matrix.

\begin{figure}[!htb]\begin{center}
%ciken 2013-Sept-10
%\includegraphics[scale=.3]{fft8.eps}
\begin{tikzpicture}
\node[vertex] (I000) at (0,3) {};
\node[vertex] (I001) at (1,3) {};
\node[vertex] (I010) at (2,3) {};
\node[vertex] (I011) at (3,3) {};
\node[vertex] (I100) at (4,3) {};
\node[vertex] (I101) at (5,3) {};
\node[vertex] (I110) at (6,3) {};
\node[vertex] (I111) at (7,3) {};
\node[vertex] (II000) at (0,2) {};
\node[vertex] (II001) at (1,2) {};
\node[vertex] (II010) at (2,2) {};
\node[vertex] (II011) at (3,2) {};
\node[vertex] (II100) at (4,2) {};
\node[vertex] (II101) at (5,2) {};
\node[vertex] (II110) at (6,2) {};
\node[vertex] (II111) at (7,2) {};
\node[vertex] (III000) at (0,1) {};
\node[vertex] (III001) at (1,1) {};
\node[vertex] (III010) at (2,1) {};
\node[vertex] (III011) at (3,1) {};
\node[vertex] (III100) at (4,1) {};
\node[vertex] (III101) at (5,1) {};
\node[vertex] (III110) at (6,1) {};
\node[vertex] (III111) at (7,1) {};
\node[vertex] (IIII000) at (0,0) {};
\node[vertex] (IIII001) at (1,0) {};
\node[vertex] (IIII010) at (2,0) {};
\node[vertex] (IIII011) at (3,0) {};
\node[vertex] (IIII100) at (4,0) {};
\node[vertex] (IIII101) at (5,0) {};
\node[vertex] (IIII110) at (6,0) {};
\node[vertex] (IIII111) at (7,0) {};
\draw (I000) -- (II001);
\draw (I001) -- (II000);
\draw (I010) -- (II011);
\draw (I011) -- (II010);
\draw (I100) -- (II101);
\draw (I101) -- (II100);
\draw (I110) -- (II111);
\draw (I111) -- (II110);
\draw (II000) -- (III010);
\draw (II001) -- (III011);
\draw (II010) -- (III000);
\draw (II011) -- (III001);
\draw (II100) -- (III110);
\draw (II101) -- (III111);
\draw (II110) -- (III100);
\draw (II111) -- (III101);
\draw (III000) -- (IIII100);
\draw (III001) -- (IIII101);
\draw (III010) -- (IIII110);
\draw (III011) -- (IIII111);
\draw (III100) -- (IIII000);
\draw (III101) -- (IIII001);
\draw (III110) -- (IIII010);
\draw (III111) -- (IIII011);
\draw (I000) -- (II000);
\draw (I001) -- (II001);
\draw (I010) -- (II010);
\draw (I011) -- (II011);
\draw (I100) -- (II100);
\draw (I101) -- (II101);
\draw (I110) -- (II110);
\draw (I111) -- (II111);
\draw (II000) -- (III000);
\draw (II001) -- (III001);
\draw (II010) -- (III010);
\draw (II011) -- (III011);
\draw (II100) -- (III100);
\draw (II101) -- (III101);
\draw (II110) -- (III110);
\draw (II111) -- (III111);
\draw (III000) -- (IIII000);
\draw (III001) -- (IIII001);
\draw (III010) -- (IIII010);
\draw (III011) -- (IIII011);
\draw (III100) -- (IIII100);
\draw (III101) -- (IIII101);
\draw (III110) -- (IIII110);
\draw (III111) -- (IIII111);
\end{tikzpicture}
\caption{\small{linear circuit for element of $FM_{8}^0$,   support is the \lq\lq butterfly graph\rq\rq }}  \label{fig:butterfly}
\end{center}
\end{figure}
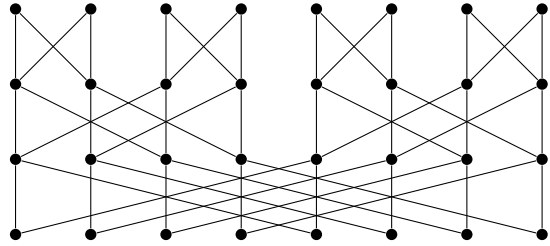

\begin{proposition}\label{factorprop}
A general factorizable matrix does not admit a linear circuit of size \mbox{$n(\tlog n+1)-1$}.
\end{proposition}
\begin{proof}
We will show that a  general  component  of $FM_n$ has dimension
$n(\tlog n+1)$, so Proposition \ref{bigvarprop} applies.

First it is clear that $\tdim FM_n$ is at most $n(\tlog n+1)$, because
$\tdim (\hat L_1\oplus\cdots \oplus \hat L_k)=2nk$ and 
if $D_1\hd D_{k-1}$ are diagonal matrices (with nonzero entries on the diagonal), then
$\Bfly(S_1D_1,D_1\inv S_2 D_2,\hd D_{k-1}\inv S_k)= \Bfly(S_1 ,  S_2  ,\hd   S_k)$,
so the fiber has dimension at least $n(k-1)$.
Consider the differential of $\Bfly$ at a general point:
\begin{align*}
d(\Bfly)|_{(S_1\hd S_k)}: \hat L_1\oplus \cdots \oplus\hat L_k &\ra Mat_{n\times n}\\
(Z_1\hd Z_k)&\mapsto Z_1S_2\cdots S_k+S_1Z_2S_3\cdots S_k+\cdots + S_1\cdots S_{k-1}Z_k
\end{align*}
The rank of this linear map is the dimension of the image of $FM_n$ as its image is
the tangent space to a general point of $FM_n$.
We may use $Z_1$ to alter $2n$ entries of the image matrix $y=S_1\cdots S_k$. Then,
{\it a priori} we could use $Z_2$ to alter $2n$ entries, but $n$ of them overlap with the
entries altered by $Z_1$, so $Z_2$ may only alter $n$ new entries. Now think of the product
of the first two matrices as fixed, then $Z_3$ multiplied by this product again can alter
$n$ new entries, and similarly for all $Z_j$. Adding up, we get $2n+(k-1)n=n(\tlog n+1)$.
\end{proof}

\section{Geometry}\label{geomsect}

\subsection{Standard facts on joins}\label{gjoinsect}
We review standard facts as well as  observations in \cite{MR2870721,LV}.
Recall the notation $\BJ(X,Y)$ from  Definition \ref{joindef}.
The following are standard facts:

\begin{proposition}\label{joinstdprop} \

\begin{enumerate}

\item \label{as0} If $X,Y$ are irreducible, then $\BJ(X,Y)$ is irreducible.
\item \label{as1}  Let $X,Y\subset \BP V$ be   varieties, then $\BI(\BJ(X,Y))\subset \BI(X)\cap \BI(Y)$.

\item\label{as2} (Terracini's Lemma) The   dimension of $\BJ(X,Y)$ is $\tdim X+\tdim Y+1-\tdim \hat T_xX\cap \hat T_yY$,
where $x\in X$, $y\in Y$ are general points. In particular,
\begin{enumerate}
\item the dimension is   $\tdim X+\tdim Y+1$   if there exist $x\in X,\ y\in Y$
such that $\hat T_xX\cap \hat T_yY=0$. ($\tdim X+\tdim Y+1$ is called the expected dimension.)

\item   If $Y=L$ is a linear space, $\BJ(X,L)$ will have
the expected dimension if and only if there exists $x\in X$ such that $\hat T_xX\cap \hat L=0$.
\end{enumerate}

\item\label{as3} If $z\in \BJ(X,p)$ and $z\not\in \langle x,p\rangle$ for some $x\in X$, then
$z$ lies on a line that is a limit of secant lines $\langle x_t,p\rangle$, for some curve $x_t$ with $x_0= p$.
\end{enumerate}

\end{proposition}
\begin{proof} For assertions \eqref{as0}, \eqref{as2}, \eqref{as3} respectively see  e.g., \cite[p157]{Harris}, \cite[p122]{MR2865915}, and
\cite[p118]{MR2865915}.
Assertion \eqref{as1} holds because $X,Y\subset \BJ(X,Y)$.
\end{proof}

To gain intuition regarding Terracini's lemma, a point on $\BJ(X,Y)$ is
obtained by selecting a point   $x\in X$ ($\tdim X$ parameters),
a point   $y\in Y$ ($\tdim Y$ parameters) and  a point on the
line joining $x$ and $y$ (one parameter). Usually these parameters
are independent, and Terracini's lemma says that if the infinitesimal
parameters are independent, the actual parameters are as well.

In the special case (3b), since $Y$ is a linear space, it is equal
to its tangent space.

To understand (4), consider   Figure \ref{figtan} where a point
on a limit of secant lines lies on a tangent line.

\begin{figure}[t]\label{figtan}
\begin{center}
\includegraphics[scale = 0.5]{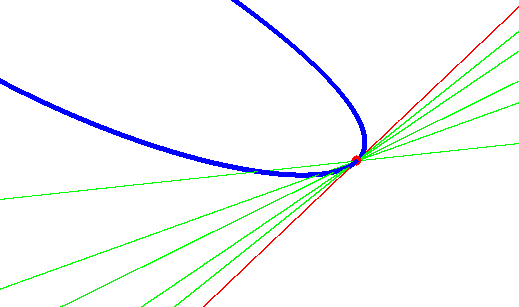}
\caption{Secant lines limiting to a tangent line}\label{tangentpic}
\end{center}
\end{figure}

\smallskip

%\section{Equations of cones}\label{joinsect}

\subsection{Ideals of cones}\label{idealofjoins}
 Define the {\it primitive part}  of the ideal of a variety $Z\subset \BP V$ as \mbox{$\BI_{prim,d}(Z):=\BI_d(Z)/(\BI_{d-1}(Z)\circ V^*)$}. 
 Here, if $A\subseteq S^{d}V$ and $B\subseteq S^{\d}V$, 
 $A\circ B:=\{pq\mid p\in A, 
 q\in B\}$.  
Note that
$\BI_{prim,d}(Z)$ is only nonzero in the degrees that minimal generators of the ideal of $Z$ appear and that
 (lifted)   bases of $\BI_{prim,d}(Z)$ for each such~$d$ furnish a set of generators of the ideal of $Z$.

\begin{proposition} \label{jjens}
 Let $X\subset \BP V$ be a variety and let $L\subset \BP V$ be a linear space.
\begin{enumerate}
\item
Then
$$
\BI_d(X)\cap S^dL\upperp\subseteq \BI_d(\BJ(X,L))\subseteq \BI_d(X)\cap (L\upperp \circ S^{d-1}V^*).
$$

\item
A set of generators
of $\BI(\BJ(X,L))$ may be taken from $\BI(X)\cap Sym(L\upperp)$.

\item In particular,  if $\BI_k(X)$ is empty, then $\BI_{k}(\BJ(X,L))$ is empty and
 $$
\BI_{k+1}(X)\cap S^{k+1}L\upperp= \BI_{k+1,prim}(\BJ(X,L))= \BI_{k+1}(\BJ(X,L)).
$$

\end{enumerate}
\end{proposition}

 Proposition \ref{jjens} says  that  we only need to look for polynomials in the variables of 
  $L\upperp$ when looking for equations of
$\BJ(X,L)$. 

\begin{proof}
For the first assertion, $P\in \BI_d(\BJ(X,L))$ if and only if $P_{k,d-k}(x,\ell)=0$ for all $[x]\in X$,  $[\ell]\in L$ and $0\leq k\leq d$ where
$P_{k,d-k}\in S^kV^*\ot S^{d-k}V^*$ is a polarization of $P$ (in coordinates,
$P_{k,d-k}$ is the coefficient of $t^k$ in the expansion of
$P(tx+y)$ in $t$, where $x,y$ are independent sets of variables
and $t$ is a single variable, see \cite[\S 7.5]{MR2865915} for more details). Now $P\in S^dL\upperp$ implies
all the terms vanish identically except for the $k=d$ term. But $P\in \BI_d(X)$ implies that term vanishes as well. The second inclusion
  of the first assertion is Proposition~\ref{joinstdprop}\eqref{as1}.

For the second assertion, we can build $L$ up by points
as $\BJ(X,\langle L',L''\rangle)=\BJ(\BJ(X,L'),L'')$, so assume $\tdim L=0$. Let $P\in \BI_d(\BJ(X,L))$. Choose a (one-dimensional) complement $W^*$ to $L\upperp$ in~$V^*$.
Write $P=\sum_{j=1}^d q_ju^{d-j}$ where $q_j\in S^jL\upperp$ and $u \in  W^*$.
Then
\begin{align}
\label{ep1}P_{j,d-j}(x^j,\ell^{d-j})&=\sum_{i=0}^j\sum_{t=0}^i (q_i)_{t,i-t}(x^t,\ell^{i-t})(u^{d-i})_{j-t,d-j+t-i}(x^{j-t},\ell^{d-j+t-i})\\
&=\label{ep2}\sum_{i=0}^j q_i(x)(u^{d-i})_{j-i,d-j}(x^{j-i},\ell^{d-j})
\end{align}
Consider the case $j=1$, then \eqref{ep2} reduces to    $q_1(x)u^{d-1}(\ell)=0$ which implies $q_1\in \BI_1(X) \cap L\upperp$.
Now consider the case $j=2$, since $q_1(x)=0$, it reduces to $q_2(x)u^{d-2}(\ell)$, so we conclude $q_2(x)\in \BI_2(X)\cap S^2L\upperp$.
Continuing, we see each $q_j \in \BI_j(X)\cap S^jL\upperp\subset \BI(\BJ(X,L))$ and the result~follows.
\end{proof}

\subsection{Degrees of  cones} \label{degsect}
For a projective variety $Z\subset \BP V$ and $z\in Z$, let
$\hat TC_zZ\subset V$ denote the {\it affine tangent cone} to $Z$ at $z$ and $TC_zZ=\BP\hat TC_zZ\subset \BP V$ the {\it (embedded) tangent
cone}. Set-theoretically $\hat TC_zZ$ is the
union of all points on all lines of the form $\tlim_{t\ra 0}\langle z,z(t)\rangle$ where $z(t)\subset \hat Z$ is a curve with
$[z(0)]=z$. If  $Z$ is irreducible, then $\tdim TC_zZ=\tdim Z$. 

To gain some intuition regarding tangent cones, we compute
the tangent cone to 
$2x^5 - w^2x^3 + w^2x^2 y + w^2x y^2 - w^2y^3 = 0$ at $[(w,x,y)] = [(1,0,0)]$.

\begin{figure}[t]
\begin{center}
\includegraphics[scale = 0.15]{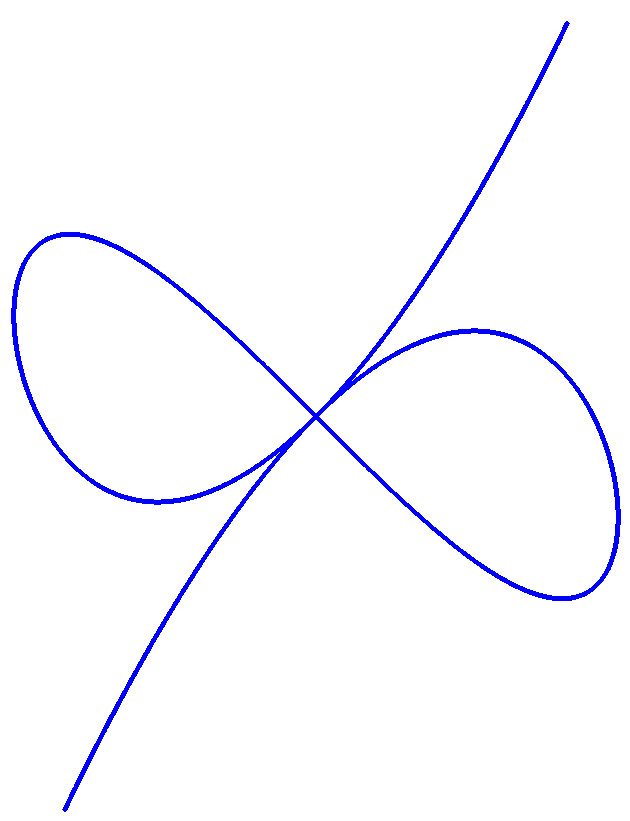} \hspace{0.2in}
\includegraphics[scale = 0.16]{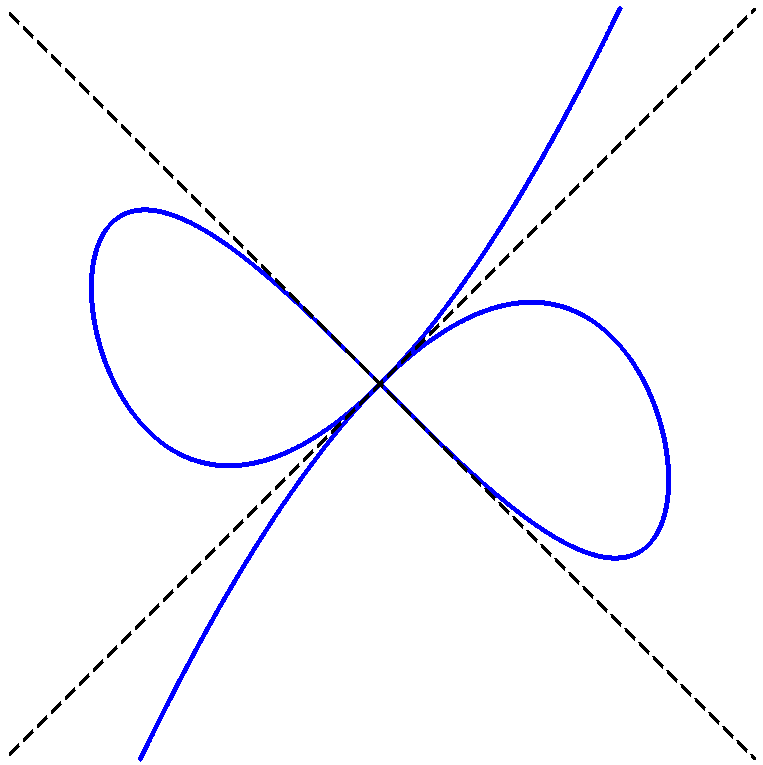} \\
(a)\hspace{1.3in}(b)
\caption{(a) is graph of   $2x^5 - x^3  + x^2 y  + x y^2  - y^3  = 0$    and (b) is the
graph with  the tangent cone at the origin.}
\label{jonpic}
\end{center}
\end{figure}

 Figure \ref{jonpic} depicts this curve in the affine space $w=1$. The
  line $\{ x+y=0\}$ has multiplicity one in the tangent
cone, and the line $\{ x-y=0\}$  has multiplicity two because two
of the branches of the curve that go through the origin are tangent to it. We will need to
keep track of these multiplicities in order to compute
  the degree of the tangent
cone as a subscheme  of the Zariski tangent space. That is,
we need to keep track of its ideal, not just its zero set.
To this end, 
 let $\fm$ denote
the maximal ideal in $\cO_{Z,z}$ of germs of regular functions on $Z$ at $z$ vanishing at $z$,
so the Zariski tangent space is $T_zZ=(\fm/\fm^2)^*$. Then the (abstract) tangent cone is the
subscheme of $T_zZ$ whose coordinate ring is the graded ring $\oplus_{j=0}^{\infty}\fm^j/\fm^{j+1}$.
 
To compute  the ideal of the tangent cone  in practice, one takes a set of generators for the ideal of $Z$ and local coordinates
$(w,y^{\a})$
such that  $z=[(1,0)]$, and writes, for each generator $P\in \BI(Z)$, $P=(w-1)^jQ(y)+ O( {(w-1)^{j+1}})$.
The generators for the ideal of the tangent cone are  the lowest degree
homogeneous components of the corresponding $Q(y)$.
See either of \cite[Ch. 20]{Harris} or
\cite[Ch. 5]{MR1344216} for details. The {\it multiplicity} of $Z$ at $z$ is
defined to be $\tmult_zZ=\tdeg(TC_zZ)$.

In our example, $(y^1,y^2)=(x,y)$ and we take coordinates with
origin at  $[(1,0,0)]$ so let $\tilde{w} = w-1$ to have the expansion $P = -\tilde{w}^2(x+y)(x-y)^2 -2 \tilde{w}(x+y)(x-y)^2 - (x+y)(x-y)^2 + 2x^5 = \tilde{w}^0[- (x+y)(x-y)^2 + 2x^5] + O(\tilde{w})$ and the ideal of the tangent cone is generated by $(x+y)(x-y)^2$. 
So in Figure \ref{jonpic}, the multiplicity at the origin  is three.
We will slightly abuse notation writing  $TC_zZ$ for both the abstract and embedded tangent  cone. 
While $TC_xZ$ may have many components with multiplicities, it is equi-dimensional, see \cite[{p162}]{MR1748380}.

\begin{proposition}\label{conedegprop}
Let $X\subset \BP V$ be a variety and let  $x\in X$. Assume that   $\BJ(X,x)\neq X$.
Let $p_x: \BP V\backslash x \ra \BP (V/\hat x)$ denote the projection
map and let $\pi:=p_x|_{X\backslash x}$. Then
$$
\tdeg(\BJ(X,x))=\frac 1{\tdeg \pi }[\tdeg(X)-\tdeg(TC_xX)].
$$
\end{proposition}

To gain intuition for Proposition \ref{conedegprop}, 
assume $\pi$ has degree one, which
it will in our situation and let $\BP W\subset \BP V$ be
a general linear space of complementary dimension to $J(X,x)$, so it intersects
$J(X,x)$ in $\tdeg(J(X,x))$ points, each of multiplicity one.  Now consider the linear space
spanned by $\BP W$ and $x$. It intersects $X$ in $\tdeg(J(X,x))+1$ points ignoring
multiplicity but it may intersect $x$ with multiplicity greater than one. 
The degree of $X$ is the number of points of intersection counted
with multiplicity, so the degree of $J(X,x)$ is the degree
of $X$ minus the multiplicity of the intersection at $x$. If $x\in X$ is
a smooth point, the multiplicity will be one, in general the multiplicity will equal
the degree of the tangent cone, which can be visualized by considering
a horizontal line through the curve in Fig. \ref{jonpic}(a), and moving the line
upwards  just
a little. The three physical points of intersection become five on the moved line.
Here is the formal proof:

\begin{proof}
By \cite[Thm. 5.11]{MR1344216},
$$
\tdeg(\ol{\pi(X\backslash x)})=\frac 1{\tdeg{\pi}}
[\tdeg(X)-\tdeg (TC_xX)].
$$
Now let $H\subset \BP V$ be a   hyperplane    not
containing $x$  that  intersects $\BJ(X,x)$ transversely. Then  $\ol{\pi(X\backslash x)}\subset \BP (V/ \hat x)$ is
isomorphic to $\BJ(X,x)\cap H\subset H$. In particular their degrees are the~same.
\end{proof}

Note that the only way to have $\tdeg(\pi)>1$ is for every secant line through $x$ to be at least a trisecant line.

\begin{proposition}\label{tconecommutea} Let $X\subset \BP V$ be a variety, let $L\subset \BP V$ be a linear space,
and let $x\in X$.
Then we have the inclusion  of schemes
\be\label{schemeincl}
\BJ(TC_xX,\tilde L)\subseteq TC_x\BJ(X,L)
\ene
where $\tilde L\subset T_x\BP V$ is the image of $L$ in the projectivized Zariski tangent space, and both
are sub-schemes of $T_x\BP V$.
\end{proposition}
\begin{proof}
Write $x=[v]$.
 For any variety $Y\subset \BP V$, generators for $TC_xY$
can be obtained from generators $Q_1\hd Q_s$ of $\BI(Y)$, see  e.g., \cite[Chap. 20]{Harris}. The generators are
  $Q_1(v^{f_1},\cdot)\hd Q_s(v^{f_s},\cdot)$ where $f_j$ is the largest  nonnegative
integer (which is at most $\tdeg Q_j-1$  since $x\in Y$) such that $Q_j(v^{f_j},\cdot)\neq 0$.
Here, if $\tdeg(Q_j)=d_j$, then  strictly speaking $Q_j(v^{f_j},\cdot)\in S^{d_j-f_j}T^*_xY$, but we
may consider $T_xY\subset T_x\BP V$ and may ignore the additional linear equations that arise as they don't effect the proof.

Generators of $\BJ(X,L)$ can be obtained from elements of $\BI(X)\cap Sym(L\upperp)$. Let
$P_1\hd P_g\in \BI(X)\cap Sym(L\upperp)$ be such a set of generators. Then, choosing the $f_j$ as
above, $P_1(v^{f_1},\cdot)\hd P_g(v^{f_g},\cdot)$
generate $\BI(TC_x(\BJ(X,L)))$.

Note that $P_1(v^{f_1},\cdot)\hd P_g(v^{f_g},\cdot)\in \BI(TC_xX)\cap Sym(L\upperp)$, so
they are in $\BI(\BJ(TC_xX,\tilde L))$.
Thus $\BI(TC_x(\BJ(X,L)))\subseteq \BI(\BJ(TC_xX,\tilde L))$.
\end{proof}

\begin{remark} The inclusion \eqref{schemeincl} may be strict. For example
$ \BJ(TC_{[a_1 \otimes b^1]}\s_r,[a_1 \otimes b^2]))\neq TC_{[a_1 \otimes b^1]}\BJ(\s_r,[a_1 \otimes b^2])$, where $a_i \otimes b^j$ is the matrix having $1$ at the entry $(i,j)$ and $0$ elsewhere. To see this,
first note that $[a_1\otimes b^2]\subset TC_{[a_1 \otimes b^1]}\s_r$, so as a set
$ \BJ(TC_{[a_1\otimes b^1]}\s_r,[a_1\otimes b^2]))= TC_{[a_1\otimes b^1]}\s_r$, in particular it
is of dimension one less than $\BJ(\s_r,[a_1\otimes b^2])$ which has the same
dimension as its tangent cone at any point.
\end{remark}

Proposition \ref{tconecommutea} implies:

\begin{corollary}\label{tconecummuteb}
Let $X\subset \BP V$ be a variety, let $L\subset \BP V$ be a linear space,
and let $x\in X$. Assume $TC_x\BJ(X,L)$ is reduced, irreducible, and $\tdim \BJ(TC_xX,\tilde L)=\tdim TC_x\BJ(X,L)$.
Then we have the equality   of schemes
$$
\BJ(TC_xX,\tilde L)=  TC_x\BJ(X,L).
$$
\end{corollary}

\subsection{Degrees  of the varieties $\BJ(\s_r,L^S)$}\label{degourvars}

\begin{lemma}\label{trisecantlemma}
Let $S$ be such that no entries of $S$ lie in a same column or row, and let $x\in S$. Assume $s<(n-r)^2$,
and let $S'=S\backslash x$. Let $\pi: \BJ(\s_r,L^{S'})\dashrightarrow \pp{n^2-2}$ denote
the projection from $[a \otimes b]$, where $a\otimes b$ is the matrix having $1$ at the entry $x$ and $0$ elsewhere. Then $\tdeg(\pi)=1$.
\end{lemma}
\begin{proof}
We need to show a general line through $[a\otimes b]$ that intersects $\BJ(\s_r,L^{S'})$, intersects it in  a unique point.
Without loss of generality, take $S$ to be the first $s$ diagonal entries and $x=x^1_1$.
It will be sufficient to show 
  that there exist  $A\in\hat \s_r$  and $M\in \hat L^{S'}$ such that  are no elements $B\in \hat\s_r$, $F\in \hat L^{S'}$
such that $u( A+M)+ v a_1 \otimes b^1 = B+F$ for some $u,v\neq 0$ other than when  $[B]=[A]$.  
Assume $A$ has no entries in the first row or column, so, moving
$F$ to the other side of the equation,
 in order that the corresponding $B$ has rank at most $r$,
there must be a matrix $D$
 with entries in $S'$, 
  such that   $A+D$ with the first row and column removed has rank at most $r-1$.

If $r\leq \lceil \frac n 2 \rceil-1$, take $A$ to be the matrix $\sum_{j=1}^{r}a_{\lfloor\frac n 2\rfloor+1+j}\otimes b^{j+1}$. 
%Then the determinant of the size $r$ submatrix
%consisting of the lower left $r\times r$ submatrix containing the truncated $A$ is always $1$.
Then the determinant of a size $r$ submatrix in the lower left quadrant of $A$ is always $1$.

If $\lceil \frac n 2 \rceil-1< r \leq n-2$, take
$A=\sum_{j=1}^{\lceil \frac n 2\rceil-1}a_{\lfloor \frac n 2 \rfloor + 1 + j} \otimes b^{j+1}
+ \sum_{i=1}^{r-\lceil \frac n 2 \rceil - 1} a_{i+1} \otimes b^{\lfloor \frac n 2 \rfloor + 1 + i}$.
Then  the size $r$ minor
consisting of columns $\{2,3,\ldots,r+1\}$ and rows
 $\{\lfloor\frac n 2 \rfloor +2,\lfloor\frac n 2 \rfloor +3, \ldots, n, 2,3,\ldots,r-\lceil\frac n 2 \rceil +2\}$ is such that
its determinant is also always $\pm 1$,   independent of choice of $D$.
\end{proof}

Let $A=\BC^n$ with basis $a_1\hd a_n$, and let $A'=\langle a_2\hd a_n\rangle$, and similarly for $B=\BC^n$. Let $x=[x^1_1]$.
It is a standard fact  (see, e.g. \cite[p 257]{Harris}), that 
$$TC_{x }\s_r(Seg(\BP A\times \BP B))=\BJ(\BP \hat T_{x}\s_1(Seg(\BP A\times \BP B)),
  \s_{r-1}(Seg(\BP A'\times \BP B'))) ,
$$
so by Proposition \ref{tconecommutea}
$$
TC_x(\BJ(\s_r(Seg(\BP A\times \BP B)),L^{S'}))
\supseteq  \BJ(\BP \hat T_{x}\s_1(Seg(\BP A\times \BP B)), \BJ(\s_{r-1}(Seg(\BP A'\times \BP B'), L^{S'})).
$$
Since $\BP \hat T_{x}\s_1(Seg(\BP A\times \BP B))$ is  a linear space  and
$\BJ(\s_{r-1}(Seg(\BP A'\times \BP B'), L^{S'})$ lies in a linear space disjoint from it,  
$$
\tdeg  \BJ(\BP \hat T_{x}\s_1(Seg(\BP A\times \BP B)), \BJ(\s_{r-1}(Seg(\BP A'\times \BP B'), L^{S'})) =
\tdeg \BJ(\s_{r-1}(Seg(\BP A'\times \BP B'), L^{S'})
$$
because if $L$ is a linear space and $Y$ any variety and $L\cap Y=\emptyset$,  then $\tdeg J(Y,L)=\tdeg Y$.

Thus if 
\be\label{samedim}
\tdim(TC_x(\BJ(\s_r(Seg(\BP A\times \BP B)),L^{S'})))=\tdim   \BJ(TC_x\s_r(Seg(\BP A\times \BP B)),L^{S'}),
\ene 
we obtain
\be\label{keydeg}
\tdeg TC_x(\BJ(\s_r(Seg(\BP A\times \BP B)),L^{S'}))\geq \tdeg \BJ(\s_{r-1}(Seg(\BP A'\times \BP B'), L^{S'}).
\ene

\smallskip

Recall the notation $d(n,r,s):=\tdeg \BJ(\s_r,L^S)$ where $S$  with $|S|=s$ is such that no  two elements lie in the same
row or column. In particular
$d(n,r,0)=\tdeg(\s_r(Seg(\pp{n-1}\times \pp{n-1}))$.

\begin{proposition}  Let $S$ be such that no two elements of $S$ lie in the  same row or column. Then
\be\label{coreqn}d_{ {n,r},s}\leq d_{ {n,r},0}-\sum_{j=1}^s d_{ {n-1,r-1},s-j}
\ene
\end{proposition}
\begin{proof} In this situation the equality \eqref{samedim} holds,  and Lemma \ref{trisecantlemma} says the degree
of $\pi$ in Proposition \ref{conedegprop} equals one, so apply it and equation \eqref{keydeg} iteratively 
to obtain   the inequalities  $d_{ {n,r},t}\leq d_{ {n,r},t-1}-d_{ {n-1,r-1},t-1}$.
\end{proof}

As mentioned in the introduction, P. Aluffi \cite{aluffideg} proved that equality holds in
\eqref{coreqn}. This has the following consequence, which was stated
as a conjecture in an earlier version of this paper:

\begin{proposition}\label{setlem} Let $S$ be such that no two elements of $S$ lie in the  same row or column and
let $x\in S$.
Then
$$
TC_x\BJ(\s_r, L^S)=\BJ(TC_x\s_r, L^{S'}).
$$
\end{proposition}

In particular $TC_x\BJ(\s_r,L^S)$ is reduced and irreducible.

\smallskip

\begin{theorem}\label{degreethm}  Each irreducible component of $\hat \cR[n,n-k,s]$ has degree at most
\be \label{amazingsum2}
\sum_{m=0}^s\binom sm (-1)^m d_{r-m,n-m,0}
\ene
with equality holding if  {no two elements of} $S$  {lie in the same row or column,} 
e.g., if the elements of $S$ appear on the diagonal.

Moreover, if we set $r=n-k$ and $s=k^2-u$ and consider the degree  {$D(n,k,u)$}
as a function of $n,k,u$, then,   fixing $k,u$ and considering $D_{k,u}(n)=D(n,k,u)$ as a function of $n$, 
it is of the form
$$
D_{k,u}(n)= (k^2)!\frac{\tb(k)^2}{\tb(2k)}p(n)
$$
where $p(n)=\frac{n^u}{u!}- \frac{k^2-u}{2(u-1)!}n^{u-1} +O(n^{u-2})$ is a polynomial of degree $u$. 
\end{theorem}

For example:
\begin{align*}
D(n,k,1)&=\frac{(k^2)!\tb(k)^2}{\tb(2k)}(n- \frac 12 (k^2-1)),  \\
D(n,k,2)&=\frac{(k^2)!\tb(k)^2}{\tb(2k)}(\frac 12 n^2- \frac 12 (k^2-2)n+\frac 16(\frac 34 k^4-\frac {11}4 k^2+2)).
\end{align*}

\begin{proof} Apply induction on all terms of \eqref{coreqn}.
We get
\begin{align*}
d_{ {n,r},s}&=d_{ {n,r},s-1}-d_{ {n-1,r-1},s-1}\\
&=(d_{ {n,r},s-2}-d_{ {n-1,r-1},s-2})-(d_{ {n-1,r-1},s-2}-d_{ {n-2,r-2},s-2})\\
&=d_{ {n,r},s-3}-3d_{ {n-1,r-1},s-3}+3d_{ {n-2,r-2},s-3}+d_{ {n-3,r-3},s-3}\\
&\vdots\\
&=
\sum_{m=0}^{\ell} (-1)^m \binom{\ell}m d_{ {n-m,r-m},s-\ell },
\end{align*}
for any $\ell \leq s$, in particular, for $\ell = s$. 
 
To see the second  assertion, note that 
$$
\sum_{m=0}^s\binom sm (-1)^m q(m) =0
$$
where $q(m)$ is any polynomial of degree less than $s$ and 
\begin{align*}
\sum_{m=0}^s\binom sm (-1)^m m^s&=s!(-1)^{s}
\\
\sum_{m=0}^s\binom sm (-1)^m m^{s+1}&=s!\binom{s-1}2(-1)^s .
\end{align*}
 (See, e.g. \cite[\S 1.4]{MR2868112},  where  the relevant  function is
called $S(n,k)$.)

Consider
\begin{align*}
\sum_{m=0}^s &\binom sm (-1)^m d_{r-m,n-m,0} \\
&=\sum_{m=0}^s\binom sm (-1)^m \frac{\tb(k)^2\tb(n-m+k)\tb(n-m-k)}{\tb(2k)\tb(n-m)^2}\\
&=\frac{ \tb(k)^2}{\tb(2k)}\sum_{m=0}^s\binom sm (-1)^m \frac{\tb(n-m+k)\tb(n-m-k)}{\tb(n-m)^2}\\
%&=\frac{ \tb(k)^2}{\tb(2k)}\sum_{m=0}^s\binom sm (-1)^m(n-k+k-1)^1(n-m+k-2)^2\cdots (n-m)^k(n-m-1)^{k-1}(n-m-2)^{k-2}\cdots
%(n-m-k+1)^1\\
&=\frac{ \tb(k)^2}{\tb(2k)}\sum_{m=0}^s\binom sm (-1)^m(n-m)^k{\textstyle \prod}_{t-1}^{k-1}(n-m+k-t)^t(n-m-k+t)^t
\end{align*}
Write $ (n-m)^k\prod_{t=1}^{k-1}(n-m+k-t)^t(n-m-k+t)^t=\sum_j c_{k,n,j}m^j$, then all values of $j$ less than $s=k^2-u$ contribute
zero to the sum,  the $j=s$ case gives $c_{k,n,k^2-u}(k^2-u)!(-1)^{k^2-u}$. Now consider  the highest power of $n$ in $c_{k,n,k^2-u}$.
$\sum_j c_{k,n,j}m^j$ is a product of  $k+2\binom k2=k^2$ linear forms, if we use $k^2-u+t$   of them for the $m$, there will
be $u-t$ to which $n$ can contribute, so the only term with  $n^u$ can come from the case $t=0$, in which case the coefficient
of $n^um^{k^2-u}$ is $(-1)^{k^2-u} \binom{k^2} u$. Putting it all together, we obtain the coefficient.
The next highest power,   $n^{u-1}$ {\it a priori} could appear  in two terms: $c_{k,n,k^2-u}$, but there the coefficient is 
$\binom{k^2}{u}[\sum_{t=1}^{k-1}(k-t)+ \sum_{t=1}^{k-1}(-k+t)]=0$,  and $c_{k,n,k^2-u+1}$, 
where the total contribution is 
$$\binom{k^2}{u-1}\binom{k^2-u+1}2 (k^2-u)!=\frac{k^2!}{(u-1)!}\frac{k^2-u}2.
$$\end{proof}

\bibliographystyle{amsplain}
 \def\cdprime{$''$} \def\cprime{$'$} \def\cprime{$'$} \def\cprime{$'$}
  \def\Dbar{\leavevmode\lower.6ex\hbox to 0pt{\hskip-.23ex \accent"16\hss}D}
  \def\cprime{$'$} \def\cprime{$'$} \def\cdprime{$''$} \def\cprime{$'$}
  \def\cprime{$'$} \def\Dbar{\leavevmode\lower.6ex\hbox to 0pt{\hskip-.23ex
  \accent"16\hss}D} \def\cprime{$'$} \def\cprime{$'$} \def\cprime{$'$}
  \def\cprime{$'$} \def\Dbar{\leavevmode\lower.6ex\hbox to 0pt{\hskip-.23ex
  \accent"16\hss}D} \def\cprime{$'$} \def\cprime{$'$}
\providecommand{\bysame}{\leavevmode\hbox to3em{\hrulefill}\thinspace}
\providecommand{\MR}{\relax\ifhmode\unskip\space\fi MR }
\providecommand{\MRhref}[2]{%
  \href{http://www.ams.org/mathscinet-getitem?mr=#1}{#2}
}
\providecommand{\href}[2]{#2}

\end{document}